\documentclass[lettersize,journal]{IEEEtran}
\usepackage{amsmath,amsthm,amssymb,amsfonts,verbatim}
\usepackage[colorlinks,
            linkcolor=blue,
            anchorcolor=blue,
            citecolor=blue
            ]{hyperref}
\usepackage[hmargin=1.2in,vmargin=1.2in]{geometry}
\usepackage{array}
\usepackage{multirow}
\usepackage{mathbbol}
\usepackage{scalerel}
\usepackage{dsfont}

\newtheorem{lemma}{Lemma}[section]
\newtheorem{theorem}[lemma]{Theorem}

\newtheorem{claim}[lemma]{Claim}
\newtheorem{prop}[lemma]{Proposition}
\newtheorem{ex}[lemma]{Example}

\newtheorem{defn}{Definition}

\theoremstyle{remark}
\newtheorem{rmk}{Remark}

\renewcommand{\epsilon}{\varepsilon}
\renewcommand{\le}{\leqslant}
\renewcommand{\ge}{\geqslant}

\newcommand{\vnote}[1]{}

\def\F{\mathds{F}}

\def \mC {\mathcal{C}}

\def \mB {\mathcal{B}}
\def \mC {\mathcal{C}}

\def \Xi {{X^{[i]}}}

\def \bc {{\bf c}}

\def \by {{\bf y}}

\usepackage{etoolbox} 
\newtoggle{comments}
\newtoggle{anonymous}
\newtoggle{short}
\newtoggle{future}
\toggletrue{anonymous}
\toggletrue{comments}
\toggletrue{short}
\togglefalse{future} 

\usepackage{todonotes}
\iftoggle{future}{
  \newcommand{\future}[1]{\todo[color=blue!25,inline]{\textsf{TODO Future:} #1}}
}{
  \newcommand{\future}[2][]{}
}

\iftoggle{comments}{%
  \newcommand{\xcp}[2][]{\todo[color=yellow!50,#1]{\textsf{XCP:} #2}}
}{
  \newcommand{\xcp}[2][]{}
}
\iftoggle{comments}{%
  \newcommand{\ivan}[2][]{\todo[color=blue!50,#1]{\textsf{ivan:} #2}}
}{
  \newcommand{\ivan}[2][]{}
}

\title{A Lower Bound on the List-Decodability of Insdel Codes}
\author{Shu Liu\thanks{Shu Liu is with the National Key Laboratory of Science and Technology on Communications, University of Electronic Science and Technology of China, Chengdu 611731, China (email: shuliu@uestc.edu.cn).},
Ivan Tjuawinata\thanks{Ivan Tjuawinata is with the Strategic Centre for Research on Privacy-Preserving Technologies and Systems, Nanyang Technological University, Singapore 637553 (email: ivan.tjuawinata@ntu.edu.sg).},
Chaoping Xing\thanks{Chaoping Xing is with the School of Electronic Information and Electric Engineering, Shanghai Jiao Tong University, Shanghai, China (email: xingcp@sjtu.edu.cn).}}
\date{}

\begin{document}
\onecolumn
\maketitle

\begin{abstract} For codes equipped with metrics such as Hamming metric, symbol pair metric or cover metric, the Johnson bound guarantees list-decodability of such codes. That is, the Johnson bound provides a lower bound on the list-decoding radius of a code in terms of its relative minimum distance $\delta$, list size $L$ and the alphabet size $q.$ For study of list-decodability of codes with insertion and deletion errors (we call such codes insdel codes), it is natural to ask the open problem whether there is also a Johnson-type bound. The problem was first investigated by Wachter-Zeh and the result was amended by Hayashi and Yasunaga where a lower bound on the list-decodability for insdel codes was derived.

The main purpose of this paper is to move a step further towards solving the above open problem. In this work, we provide a new lower bound for the list-decodability of an insdel code. As a consequence, we show that unlike the Johnson bound for codes under other metrics that is tight, the bound on list-decodability of insdel codes given by Hayashi and Yasunaga is not tight. Our main idea is to show that if an insdel code with a given Levenshtein distance $d$ is not list-decodable with list size $L$, then the list decoding radius is lower bounded by a bound involving $L$ and $d$. In other words, if the list decoding radius is less than this lower bound, the code must be list-decodable with list size $L$. At the end of the paper we use such bound to provide an insdel-list-decodability bound for various well-known codes, which has not been extensively studied before.
\end{abstract}

\section{Introduction}\label{sec:intro}

The Johnson bound is a benchmark for the study of list-decodability of codes. The Johnson bounds for Hamming, symbol-pair and cover metric codes were derived in \cite{Joh62,Joh63}, \cite{LXY18} and \cite{LXY19}, respectively. The usual way to derive the Johnson bound with a given minimum distance is to show that the intersection of a code of the given minimum distance with any ball of radius $t$ contains at most $L$ elements. However, for some metrics such as rank-metric, a Johnson bound does not exist \cite{Wac13}.

When considering Levenshtein distance, which is measured based on the minimum amount of insertion and deletion operation required to transform one word to another, the classical approach does not work well. This is due to the fact that, in contrast to the Hamming distance or other distances, which is invariant of translation, the Levenshtein distance is not. This causes the size of a Levenshtein ball of any positive radius to be dependent of the centre. Hence, one wonders whether there is also the  Johnson-type bound on list-decodability of insdel codes. On the other hand, the Johnson bounds for Hamming, symbol-pair  and cover metric codes  obey the following two properties:
 \begin{itemize}
 \item[(i)] every code that satisfies such a bound is guaranteed to be list-decodable with polynomial list size;
 \item[(ii)] For any decoding radius exceeding the bound, there exists a code that is not list-decodable up to such radius for any polynomial list size. 
\end{itemize}

In this paper, we define the bound satisfying the above two properties to be the Johnson bound. For a more formal definition of Johnson bound, interested readers can refer to \cite[Chapter $4$]{Gur05}. In particular, by \cite{Wac13}, the Johnson bound for rank-metric codes is just half of the minimum distance which is also the unique decoding radius. It is clear that for study of list-decodability, it is of great importance to derive the Johnson bound satisfying the above properties. Thus, we propose an open problem.

\noindent\textbf{Open problem:} Does the Johnson-type bound exist for insdel codes that satisfies the  two requirements discussed above? If it exists, find the bound.

The main purpose of this paper is to investigate lower bounds on the list-decodability of insdel codes in the effort to provide a clearer picture towards the Johnson-type bound we have previously discussed.

\subsection{Informal definition and brief literature review}

Insertion and deletion (Insdel for short) errors are synchronization errors~\cite{HS17,HSS18} in communication systems caused by the loss of positional information of the message. They have recently attracted many attention due to their applicabilities in many interesting fields such as DNA storage and DNA analysis~\cite{JHS+17, XW05}, race-track memory error correction \cite{CKV+17} and language processing \cite{BM00, Och03}.

The study of codes with insertion and deletion errors was pioneered by Levenshtein, Varshamov and Tenengolts in the 1960s~\cite{VT65, Lev65, Lev67, Ten84}. This study was then further developed by Brakensiek, Guruswami and Zbarsky \cite{BGZ18}. There have also been different directions for the study of insdel codes such as the study of some special forms of the insdel errors \cite{SWG+17, CSF+14, LWY17, Mit08} as well as their relations with Weyl groups \cite{Hag18}.

Given two words defined over the same alphabet (not necessarily of the same length), we define their Levenshtein distance to be the minimum number of insertion and deletion operations required to transform one word to the other\footnote{Originally, Levenshtein distance is defined to be the minimum number of synchronization operation, including substitution, required to transform one word to the other. In this work, we abuse the term for ease of notation, as remarked in Remark~\ref{rmk:Levenshtein}.}. As in the other metrics, the unique decoding radius  of an insdel code is completely determined by its  Levenshtein distance. Precisely speaking, the total number of insertion and deletion errors that can be uniquely corrected is $\lfloor(d-1)/2\rfloor$  if Levenshtein distance is $d$.
 However, in the case of list-decoding, we have to separate insertion and deletion errors as explained in the Subsection \ref{sec:IntPrevRes}.
Thus, we call a code of length $n$ imbued with Levenshtein distance to be $(\tau_I,\tau_D,L)$-insdel-list-decodable if it can list-decode against up to $\tau_I$ insertions and $\tau_D$ deletion errors with list size of at most $L.$ The formal definitions of Levenshtein distance and insdel-list-decodability can be found in Definitions~\ref{def:insdeldist} and \ref{def:LD} respectively.



\subsection{Previous Results}\label{sec:IntPrevRes}
There have been several works on the list-decoding of codes under Levenshtein distance. Guruswami and Wang \cite{GW17} studied list-decoding of binary codes with deletion errors only. They constructed binary codes with decoding radius close to $\frac 12$ for deletion errors only. Based on the indexing scheme and concatenated codes, there are further works that provided efficient encoding and decoding algorithms by concatenating an inner code achieving a previously derived bound and an outer list-recoverable Reed-Solomon code achieving the classical Johnson bound.
In 2018, Haeupler~\textit{et al.}~\cite{HSS18} provided an explicit construction of a family of list-decodable insdel codes through the use of synchronization strings with large list-decoding radius for sufficiently large alphabet size and designed its efficient list-decoding algorithm. Furthermore, they derived some upper bounds on list-decodability  for insertion or deletion errors. They also established that in contrast to the unique decoding scenario where the effect of insertion and deletion errors are equivalent, in the list-decoding scenario, insertion errors have less effect on the insdel-list-decodability compared to deletion errors. It can be observed that the maximum fraction of deletion errors that can be list-decoded by a code of rate $R$ is at most $1-R<1$~\cite[Theorem 1.3]{HSS18}. On the other hand, any amount of insertion errors can always be resolved if the code is defined over a sufficiently large alphabet. This observation shows that in an investigation of list-decoding of insdel codes, the number of insertion errors should be separated from the number of deletion errors. Lastly, they considered the list-decodability of random codes with insertion or deletion errors only. Their results reveal that there is a gap between the upper bound on list-decodability of insertion (or deletion) codes and list-decodability of a random insertion (or deletion) code.

Liu \textit{et al.} \cite{LTX21} investigated the list-decodability of a random code of a given rate. Furthermore, with the help of concatenation and indexing scheme, they have also provided a Monte-Carlo construction of an insdel-list-decodable code. Haeupler, Rubinstein and Shahrasbi~\cite{HRS19} introduced probabilistic fast-decodable indexing schemes for insdel distance which reduces the computing complexity of the list-decoding algorithm in~\cite{HSS18}. In 2020, Guruswami \textit{et al.}~\cite{GHS20} established the zero-rate threshold of insdel-list-decodable codes. More specifically, for any alphabet size $q,$ they establish the set of all possible amount of insertion and deletion errors such that there exists a code of positive rate that can list-decode those amounts of errors. Recently, Haeupler and Shahrasbi~\cite{HS20} expanded the bound derived in~\cite{GHS20} to provide an upper bound of the list-decodability of insdel codes.

In general the works on list-decoding of insdel codes we have discussed above focus on one of the following two directions. Firstly, some of the works focuses on construction of specific insdel-list-decodable code with a given rate. Secondly, some other works focus on finding the upper bound of the insdel-list-decodability of a code, which is a necessary condition for a code to be insdel-list-decodable. Although such bounds provide the largest possible parameter for a code that can list-decode a given number of insertion and deletion errors, it does not guarantee that any code of a given parameter to be insdel-list-decodable. Hence, given a code of a given parameter, it is not guaranteed to be insdel-list-decodable and there needs to be a scheme to check whether it is insdel-list-decodable. This implies that such works do not provide any clearer picture about the Johnson-type bound we are aiming to obtain.
On the other hand, there have also been a few works on the lower bound of the insdel-list-decodability of a code. Such works provides a lower bound to the Johnson-type bound. More specifically, it satisfies the first condition of the Johnson-type bound that we have discussed above. However, it does not guarantee that any code that does not satisfy those bounds cannot be list-decodable for any polynomial list size.

In 2017, Wachter-Zeh~\cite{Wac18} firstly considered the list-decoding of insdel codes and provided a lower bound of the insdel-list-decodability of a code. Hayashi and Yasunaga~\cite{HY20} provided some amendments on the result in~\cite{Wac18} and derived a lower bound which is only meaningful when insertion occurs. Unlike the Johnson bound for Hamming-metric codes, the bound given in \cite{HY20} is not tight (see Subsection \ref{sec:IntComp}). 
Thus, it is interesting to see if we can improve such lower bounds to obtain a closer estimate to the true Johnson-type bound satisfying both conditions discussed above.


\subsection{Our Main Contribution}\label{sec:IntOurRes}

Following the work of Wachter-Zeh~\cite{Wac18} which was then amended by Hayashi and Yasunaga~\cite{HY20}, we focus on the lower bound of insdel-list-decodability which provides bounds on parameters that guarantees that a code with the given parameters is insdel-list-decodable. More specifically, we focus on a sufficient condition on the relative minimum Levenshtein distance of a code for it to be insdel-list-decodable. First, we provide an informal restatement of our main result on the sufficient condition for a code to be insdel-list-decodable.

\begin{theorem}[Informal Restatement of Theorem~\ref{thm:genl}]
Let $\mathcal{C}$ be a code of length $n$ over an alphabet of size $q$ with minimum Levenshtein distance $d=2\delta n$ and $L\ge 2$ be a positive integer. Suppose that $\tau_I<q-1$ and $\tau_D<\frac{q-1}{q}$ are non-negative real numbers. If $\tau_D<\delta$ and $\tau_I<\rho^{(\delta,L)}(1-\tau_D)$ where
\[\rho^{(\delta,L)}(x)=\max_{r=1,\cdots, L}\left\{\frac{2L-r+1}{L+1} x -\frac{L}{r}(1-\delta)\right\},\]
then $\mathcal{C}$ is $(\tau_I,\tau_D,L)$-insdel-list-decodable.
\end{theorem}

Here, $\rho^{(\delta,L)}(1-\tau_D)$ is a piecewise linear function with $L$ linear pieces where it coincides with the unique decoding bound $\tau_I<\delta-\tau_D$ when $\tau_D\ge 1-\frac{L+1}{L-1}(1-\delta).$

Such result provides a lower bound on the insdel-list-decodability of a code given the values of $\delta, \tau_I, \tau_D$ and $L.$ As illustrated in Section~\ref{sec:inscons}, Theorem~\ref{thm:genl} provides an insdel-list-decodability results for various constructions of codes, including some Reed-Solomon codes~\cite{DLTX21, LT21, CZ22, LX21, CST21}, Varshamov-Tenengolts (VT for short) codes~\cite{VT65, Lev65, Ten84} and Helberg codes~\cite{HF02, LN16}. This provides such constructions with a stronger property on their insdel-list-decodability, as can be observed in Theorems~\ref{thm:RSLD}, \ref{thm:VTLD} and \ref{thm:HCLD} respectively. As insdel-list-decodability property has not been studied for such codes except for a brief analysis for VT codes in \cite{Wac18}, an interesting open question to consider is to provide an efficient insdel-list-decoding algorithm for all these codes for a general list size $L.$

Another potential application of such result is in the construction of insdel-list-decodable codes with efficient list-decoding algorithms (see for example~\cite{GL16, GW17, HSS18, LTX21}). One of the common methods for constructing such code is to use an insdel-list-decodable inner code with short length, equipped with either an indexing scheme or synchronization strings, which is then concatenated with an outer code that is either list-recoverable under the classical Hamming metric (see for example \cite{GR08,GX13,HRW17}) or list-decodable against erasure errors (See for example \cite{Gur03,GR06,BDT18}). Such concatenation allows the decoding algorithm to be designed by using exhaustive search to list decode the inner code, which transforms the problem to either a list-recovery problem or list decoding against erasures, which can then be solved using the list recovery algorithm of the outer code. In such codes, the inner code is in general either specifically constructed or is sampled at random with the condition that the sampled code is list-decodable with high probability. Such approaches either requires a specific code for the inner code or an additional scheme to verify that the sampled code is indeed list-decodable. Having the result presented in Theorem~\ref{thm:genl}, we may sample any insdel code with a given minimum Levenshtein distance and we can guarantee that such code is insdel-list-decodable.

\subsection{Our Technique}\label{sec:IntOurTec}
In order to show our main claim, we suppose that the conditions of $\tau_D$ and $\tau_I$ are satisfied and $\mathcal{C}$ is not $(\tau_I,\tau_D,L)$-insdel-list-decodable. Then, there exists an integer $N\in[n-\tau_D n, n+\tau_I n]$ and a word $\mathbf{y}$ of length $N$ such that there exists $L+1$ codewords $\mathbf{c}_0,\cdots, \mathbf{c}_{L}$ that can be obtained from $\mathbf{y}$ by $\tau_D n$ insertion and $\tau_I n$ deletion operations. For $i=0,\cdots, L,$ define $Y_i\subseteq\{1,\cdots, N\}$ the set of indices that correspond to a longest common subsequence of $\mathbf{y}$ and $\mathbf{c}_i.$ It is then easy to see that $Y_i\cap Y_j$ is a subset of the indices that correspond to a longest common subsequence of $\mathbf{c}_i$ and $\mathbf{c}_j$ seen as entries of $\mathbf{y}.$ Then, due to the relation of $Y_i, \tau_I$ and $\tau_D$ as well as the relation between $Y_i\cap Y_j$ with $\delta,$ any relation we may derive on the sets $Y_0,\cdots, Y_L$ and their intersections or unions also provides a relation between the code parameters, $\delta, \tau_I,\tau_D$ and $L.$ More concretely, we examine the size of the unions of sets obtained by intersecting a fixed number of $Y_i.$ That is, for various values of $r=1,\cdots, L,$ we consider $\left|\bigcup_{0\le i_1<\cdots<i_r\le L} \bigcap_{j=1}^r Y_{i_j}\right|.$ With the help of inclusion-exclusion principle as well as some carefully chosen linear combination of the analysis for different values of $r,$ we obtain the relation $\tau_I\ge\rho^{(\delta,L)}(1-\tau_D)$ which contradicts our condition on $\tau_I.$

 We conclude this section by briefly discussing our strategy on choosing the linear combination to obtain the $L$ bounds that defines our bound $\rho^{(\delta,L)}(x).$ Following our notations in Section~\ref{sec:InsL}, for any positive integer $L\ge 2$ and $j\le L+1,$ define $\Sigma_j^{(L+1)}\triangleq \sum_{0\le i_1<i_2<\cdots<i_j\le L}|\bigcap_{\ell=1}^{j} Y_{i_\ell}|$ and $\Psi_{j}^{(L+1)}\triangleq \left|\bigcup_{0\le i_1<i_2<\cdots<i_j\le L}\bigcap_{\ell=1}^{j} Y_{i_\ell}\right|.$ It is easy to see that for any $v=1,\cdots, L+1,$ by Inclusion-Exclusion Principle, we can express $\Psi_v^{(L+1)}$ as a linear combination of $\Sigma_v^{(L+1)},\cdots, \Sigma_{L+1}^{(L+1)},$ i.e., there exist constants $A_{v,v},\cdots, A_{L+1,v}$ such that $\Psi_v^{(L+1)}=\sum_{j=v}^{L+1} A_{j,v} \Sigma_j^{(L+1)}.$ Hence, for any real numbers $c_1,\cdots, c_{L+1},$ we have $\sum_{v=1}^{L+1} c_v \Psi_v^{(L+1)}=\sum_{v=1}^{L+1} c_v\sum_{j=v}^{L+1} A_{j,v}\Sigma_j^{(L+1)} = \sum_{j=1}^{L+1} \Phi_{j}^{(L+1)} \Sigma_j^{(L+1)}$ for some constants $\Phi_1^{(L+1)},\cdots, \Phi_{L+1}^{(L+1)}.$ On one hand, we note that $\Psi_v^{(L+1)}\leq N$ for any $v=1,\cdots,L+1.$ Hence $\sum_{v=1}^{L+1} c_v \Psi_v^{(L+1)}$ is upper bounded by $\sum_{v=1}^{L+1} c_v N.$ On the other hand, for $j=1,2,$ we can obtain a bound on $\Sigma_j^{(L+1)}$ based on its relation with the assumption of $\mathcal{C}$ being not insdel-list-decodable and its relative minimum Levenshtein distance respectively. Hence, assuming that $\sum_{j=3}^{L+1} \Phi_j^{(L+1)} \Sigma_j^{(L+1)}\ge 0, \sum_{j=1}^{L+1} \Phi_{j}^{(L+1)} \Sigma_j^{(L+1)}$ can be lower bounded by some function of $\tau_I, \tau_D, \delta, L$ and $n.$ In conclusion, for any choice of such $c_1,\cdots, c_{L+1}$ such that $\sum_{j=3}^{L+1} \Phi_j^{(L+1)} \Sigma_j^{(L+1)}\ge 0,$ recalling that $n-\tau_D n\le N\le n+\tau_I n,$ we obtain an inequality on $\tau_I,\tau_D,\delta$ and $L.$ Here, we choose such $c_1,\cdots, c_{L+1}$ such that we may eliminate the terms corresponding to $\Sigma_j^{(L+1)}$ for $j\ge 3.$ Note that by definition, we have $\Sigma_j^{(L+1)}\ge \Sigma_{j+1}^{(L+1)}$ for $j\le L.$ Hence, to reduce the impact of using the inequality, $\sum_{j=3}^{L+1} \Phi_j^{(L+1)} \Sigma_j^{(L+1)}\ge 0,$ we choose $c_1,\cdots, c_{L+1}$ that eliminate $\Sigma_j^{(L+1)}$ in an ascending order on the value of $j.$ Hence, we are left with finding the values of $A_{j,v}$ for different values of $j$ and $v.$ We utilize the notion of $v$-cover of a set to express each coefficient as a function of the number of $v$-cover of different sets, denoted by $A_{j,\ell,v}.$ Intuitively, the term $v$-cover refers to the the covering of the set with its subsets, each of size $v.$ Utilizing a recursive relation we derive for $A_{j,\ell,v},$ we can find an explicit formula for $A_{j,v}.$ A more detailed discussion in the value of $A_{j,v}$ along with a formal definition of $v$-cover and the choice of the constant $c_i$'s can be found in Section~\ref{sec:InsL}. The general strategy is also discussed in more detail in Remark~\ref{rmk:strat}.

In our work, we only consider $L$ specific combinations of the sizes of such unions. It is interesting to see if a better bound can be derived by considering a different combination of the unions. Similarly, such analysis can also be used to analyse the size of different combination of $Y_i$'s which is different from just unions of sets obtained by intersecting the same number of sets.

\subsection{Comparison}\label{sec:IntComp}
In this work, we compare our result with the unique decoding bound $\tau_I+\tau_D<\delta$ and the bound by Hayashi and Yasunaga~\cite{HY20} which we denote by HY bound. The comparison with unique decoding bound yields the fact that due to the condition $\tau_D<\delta,$ similar to the HY bound, our bound is only meaningful when $\tau_I>0.$ Furthermore, based on the form of our bound, when $\tau_D<1-\frac{L+1}{L-1}(1-\delta),$ our bound outperforms the unique decoding bound. On the other hand, we also show that for sufficiently large $\delta,$ there exists a range of $\tau_D$ such that in such values of $\tau_D,$ our bound also outperforms the HY bound. A formal discussion of such comparison can be found in Section~\ref{sec:Inscom}. An illustration of this result for $L=2$ and $\delta=0.9$ can be found in Figure~\ref{fig:rhodelta2exalt}. Note that the value of $q$ only affects the curves in Figure~\ref{fig:rhodelta2exalt} in determining the curves' endpoints while the overall curves do not change. Because of this, Figure~\ref{fig:rhodelta2exalt} is drawn without considering the effect of $q.$ We define $P_1,$ the point with smallest value of $\tau_D$ such that our bound outperforms the HY bound and $P_2,$ the point where our bound coincides with the unique decoding bound.

\begin{figure}[ht]
     \centering

\includegraphics[scale=0.75]{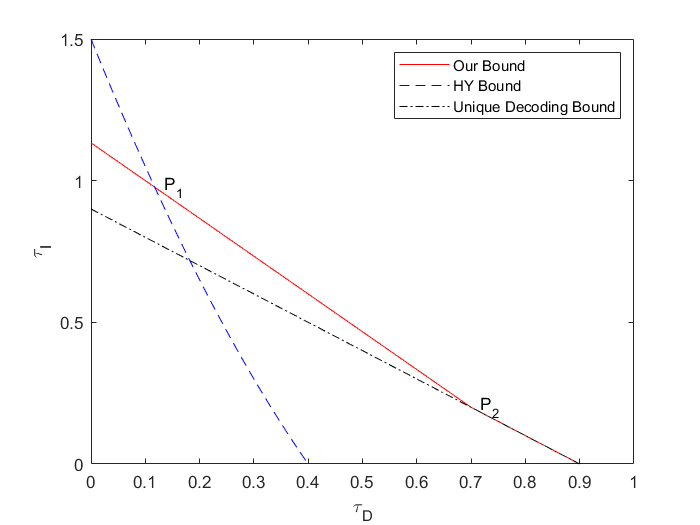}

\caption{Comparison between Our Bound, HY Bound~\cite{HY20} and Unique Decoding Bound when $L=2$ and $\delta=0.9.$}
\label{fig:rhodelta2exalt}
\end{figure}

\subsection{Organization}\label{sec:IntOrg}
The rest of this paper is organized as follows. In Section~\ref{sec:prelim}, we briefly introduce some basic definitions and results that are essential in our discussion. We consider a special case of our main result when the list size is $2$ in Section~\ref{sec:Insdel2} for a simple illustration of our approach. Such approach is then generalized for any list size in Section~\ref{sec:InsL}. We compare and prove that our result contains an improvement to previously established lower bound of insdel-list-decodability in Section~\ref{sec:Inscom}. Lastly, we provide some insdel-list-decodability result of various codes in Section~\ref{sec:inscons}.

\section{Preliminaries}\label{sec:prelim}
For positive integers $q$ and $n$ such, let $\mathbb{\Sigma}_q$ be a finite alphabet of size $q$ and $\mathbb{\Sigma}_q^n$ be the set of all vectors of length $n.$  For any positive real number $i,$ we denote by $[i],$ the set of integers $\{1,\cdots, \lfloor i\rfloor\}.$ Given a vector $\mathbf{v}=(v_1,\cdots, v_n)\in \mathbb{\Sigma}_q^n$ and a set  $S\subseteq [n],$ we define the projection of $\mathbf{v}$ at $S,$ denoted by $\left.\mathbf{v}\right|_S$ as a vector of length $|S|$ containing the coordinates of $\mathbf{v}$ with indices in $S.$ That is, $ \left.\mathbf{v}\right|_S=(v_i)_{i\in S}.$ Given $\alpha\in \mathbb{\Sigma}_q$ and any non-negative integer $m,$ we define ${\boldsymbol \alpha}^m\in \mathbb{\Sigma}_q^m$ obtained by repeating $\alpha~m$ times. Furthermore, for any two vectors $\mathbf{u}\in \mathbb{\Sigma}_q^{n_1}$ and $\mathbf{v}\in\mathbb{\Sigma}_q^{n_2}$ over $\mathbb{\Sigma}_q,$ we define $(\mathbf{u}\|\mathbf{v})\in \mathbb{\Sigma}_q^{n_1+n_2}$ to be the vector obtained by concatenating $\mathbf{v}$ in the right of $\mathbf{u}.$

\begin{defn}\label{def:insdeldist}
Let $\mathbf{a}\in \mathbb{\Sigma}_q^{n_1}$ and $\mathbf{b}\in\mathbb{\Sigma}_q^{n_2}$ be two words over $\mathbb{\Sigma}_q$ not necessarily of the same length.
\begin{enumerate}
\item We define the Levenshtein distance between $\mathbf{a}$ and $\mathbf{b}, d_{L}(\mathbf{a},\mathbf{b})=d$ if there exists non-negative integers $t_I$ and $t_D$ such that $t_I+t_D=d$ and we may obtain $\mathbf{b}$ from $\mathbf{a}$ by inserting $t_I$ symbols and deleting $t_D$ symbols. Furthermore, we also require that for any $t_I'$ and $t_D'$ such that $t_I'+t_D'<d,$ it is impossible to obtain $\mathbf{b}$ from $\mathbf{a}$ by inserting $t_I'$ symbols and deleting $t_D'$ symbols.
\item For a non-negative integer $m\le \min\{n_1,n_2\},$ we say $\mathbf{v}=(v_1,\cdots, v_{m})$ is a common subsequence of $\mathbf{a}$ and $\mathbf{b}$ if there exists $1\le i_1<i_2<\cdots<i_m\le n_1$ and $1\le j_1<j_2<\cdots<j_m\le n_2$ such that $\mathbf{v}=\mathbf{a}|_{\{i_1,\cdots,i_m\}}=\mathbf{b}|_{\{j_1,\cdots,j_m\}}.$ Furthermore, we say that $\mathbf{v}$ is a longest common subsequence of $\mathbf{a}$ and $\mathbf{b}$ if there does not exist $\mathbf{u}$ of length $m'>m$ such that $\mathbf{u}$ is also a common subsequence of $\mathbf{a}$ and $\mathbf{b}.$ We denote by $\ell_{\mathtt{LCS}}(\mathbf{a},\mathbf{b})$ the length of a longest common subsequence of $\mathbf{a}$ and $\mathbf{b}.$
\end{enumerate}

\end{defn}

\begin{rmk}\label{rmk:Levenshtein}
We note that the Levenshtein distance of two words was defined to be  the minimum number of insertions, deletions and substitutions required to transform a word to the other, which was first defined by Levenshtein~\cite{Lev65,Lev67}. In this manuscript, to simplify the notation, we will abuse the term and call $d_L(\mathbf{a},\mathbf{b})$ the Levenshtein distance. This is to differentiate it from the term ``insdel'' that we use when we want to separate the number of insertions from deletions.
\end{rmk}

\begin{rmk}\label{rmk:insdeldistsamelen}
Let $\mathbf{a}\in \mathbb{\Sigma}_q^{n_1}$ and $\mathbf{b}\in \mathbb{\Sigma}_q^{n_2}$ be two words over $\mathbb{\Sigma}_q$ not necessarily of the same length. Assuming that $\mathbf{v}$ is a longest common subsequence of $\mathbf{a}$ and $\mathbf{b}$ of length $\ell=\ell_{\mathtt{LCS}}(\mathbf{a},\mathbf{b}),$ we can obtain $\mathbf{b}$ from $\mathbf{a}$ by deleting all $n_1- \ell$ symbols in $\mathbf{a}$ outside $\mathbf{v}$ and inserting all $n_2-\ell$ symbols of $\mathbf{b}$ outside $\mathbf{v}$ and it can also be verified that $d_L(\mathbf{a},\mathbf{b})=n_1+n_2-2\ell.$ Since $0\le \ell\le \min\{n_1,n_2\},$ we have $|n_1-n_2|\le d_L(\mathbf{a},\mathbf{b})\le n_1+n_2.$

Note that if $d_L(\mathbf{a},\mathbf{b})=d$ where $\mathbf{a}$ and $\mathbf{b}$ have the same length, we must have $d$ to be an even integer.
\end{rmk}

In the following we provide an alternative definition of Levenshtein distance with respect to the numbers of insertion and deletion required

\begin{prop}\label{prop:2insdeldist}
Let $\mathbf{a}\in \mathbb{\Sigma}_q^{n_1}$ and $\mathbf{b}\in\mathbb{\Sigma}_q^{n_2}$ be two words over $\mathbb{\Sigma}_q.$ Then $d_L(\mathbf{a},\mathbf{b})=d$ if and only if there exist $t_I$ and $t_D$ such that \begin{enumerate}
\item $t_I+t_D=d.$
\item We can obtain $\mathbf{b}$ from $\mathbf{a}$ by inserting $t_I$ symbols and deleting $t_D$ symbols
\item For any pair $(t_I',t_D')$ such that $\mathbf{b}$ can be obtained from $\mathbf{a}$ by inserting $t_I'$ symbols and deleting $t_D'$ symbols, then $t_I'\ge t_I$ and $t_D'\ge t_D.$
\end{enumerate}
\end{prop}
\begin{proof}
First, suppose that $d_L(\mathbf{a},\mathbf{b})=d.$ Then by definition, there exists non-negative integers $t_I$ and $t_D$ such that $t_I+t_D=d$ and we may obtain $\mathbf{b}$ from $\mathbf{a}$ by inserting $t_I$ symbols and deleting $t_D$ symbols. This directly proves the first two claims. Suppose that there exists a pair $(t_I',t_D')$ such that $t_I'< t_I$ or $t_D'< t_D$ and we may obtain $\mathbf{b}$ from $\mathbf{a}$ by inserting $t_I'$ symbols and deleting $t_D'$ symbols. Note that we must have $t_I-t_D=n_2-n_1=t_I'-t_D'.$ Hence $t_I'<t_I$ if and only if $t_D'<t_D.$ In such case, $t_I'+t_D'<d,$ contradicting the assumption that $d_L(\mathbf{a},\mathbf{b})=d.$ Hence we must have $t_I'\ge t_I$ and $t_D'\ge t_D.$

Now suppose that there exist $t_I$ and $t_D$ satisfying the three conditions. We aim to show that $d_L(\mathbf{a},\mathbf{b})=t_I+t_D.$ By the second assumption, we have $d_L(\mathbf{a},\mathbf{b})\le t_I+t_D.$ Now suppose that there exists $t_I'$ and $t_D'$ such that $t_I'+t_D'<t_I+t_D$ and we can obtain $\mathbf{b}$ from $\mathbf{a}$ by inserting $t_I'$ symbols and deleting $t_D'$ symbols. However, by the third assumption, we must have $t_I'\ge t_I$ and $t_D'\ge t_D$ which implies $t_I'+t_D'<t_I+t_D\le t_I'+t_D',$ which is a contradiction. Then we must have $t_I'+t_D'\ge t_I+t_D,$ completing the proof.
\end{proof}

A code $\mathcal{C}$ over $\mathbb{\Sigma}_q$ of length $n$ is a non-empty subset $\mathbb{\Sigma}_q^n.$ We define the rate of $\mathcal{C}$ to be $\mathcal{R}(\mathcal{C})\triangleq \frac{\log_q|C|}{n}.$ The minimum Levenshtein distance of $\mathcal{C}$ is defined to be $d_L(\mathcal{C})\triangleq \min_{\mathbf{a},\mathbf{b}\in \mathcal{C},\mathbf{a}\neq \mathbf{b}}\{d_L(\mathbf{a},\mathbf{b})\}.$ We call a code that we consider under Levenshtein distance as an insdel code.
We define the relative minimum Levenshtein distance of $\mathcal{C}$ to be $\delta_L(\mathcal{C})=\frac{d_L(\mathcal{C})}{2n}.$

Next we discuss the definition of Levenshtein and insdel balls.

\begin{defn}\label{def:balls}
Let $q\ge 2,n\ge 2$ and $d$ be positive integers and $t_I$ and $t_D$ be non-negative integers such that $t_D\le n.$
\begin{enumerate}
\item For any $\mathbf{x}\in \mathbb{\Sigma}_q^n,$ we define the Levenshtein ball with centre $\mathbf{x}$ and radius $d$ as $\mathcal{B}_L(\mathbf{x},d)\triangleq\{\mathbf{y}\in \mathbb{\Sigma}_q^\ast: d_L(\mathbf{x},\mathbf{y})\le d\}.$
\item For any $\mathbf{x}\in \mathbb{\Sigma}_q^n,$ we define the insdel ball with centre $\mathbf{x}$ and insertion radius $t_I$ and deletion radius $t_D$ as $\mathcal{B}_{ID}(\mathbf{x},t_I,t_D)\triangleq\{\mathbf{y}\in \mathbb{\Sigma}_q^\ast: \exists 0\le t_I'\le t_I, 0\le t_D'\le t_D, \mathbf{y}\mathrm{~can~be~obtained~from~}\mathbf{x}\mathrm{~by~inserting~}t_I'$ $\mathrm{~symbols~and~deleting~}t_D'$ $\mathrm{symbols}\}.$
\end{enumerate}
\end{defn}

We can then use the Levenshtein and insdel balls to define the list-decodability of a code.

\begin{defn}\label{def:LD}
Let $\mathcal{C}\subseteq \mathbb{\Sigma}_q^n$ be a code and let $\tau,\tau_I,\tau_D\ge 0$ be non-negative real numbers and $L\ge 1$ be integers. We say $\mathcal{C}$ is $(\tau,L)$-Levenshtein-list-decodable if for any non-negative integer $N$ such that $n-\tau n\le N\le n+\tau n$ and every $\mathbf{r}\in\mathbb{\Sigma}_q^N,$ we must have $|\mathcal{B}_L(\mathbf{r},\tau n)\cap \mathcal{C}|\le L.$ Furthermore, $\mathcal{C}$ is said to be $(\tau_I,\tau_D,L)$-insdel-list-decodable if for any non-negative integer $N$ such that $n-\tau_D n\le N\le n+\tau_I n$ and every $\mathbf{r}\in \mathbb{\Sigma}_q^N,$ we must have $|\mathcal{B}_{ID}(\mathbf{r},\tau_D n,\tau_I n)\cap \mathcal{C}|\le L.$
\end{defn}

In general, we want $\lim_{n\rightarrow \infty} \mathcal{R}(\mathcal{C})>0$ and $L=\mathrm{poly}(n).$ Note that by allowing $\tau_D n\ge\frac{(q-1)}{q} n$ deletions we can always transform any codeword to a word of length $\frac{n}{q}$ with all entries being the most frequently occurring element in $\mathbf{c}.$ Hence if $\mathcal{C}$ is list-decodable against $\tau_D n$ deletions with list size $L=\mathrm{poly}(n),$ there can only be at most $Lq$ codewords which means that $\mathcal{R}(\mathcal{C})$ tends to $0.$ Hence to have a positive rate and polynomial list size, $\tau_D<\frac{(q-1)}{q}.$ Similarly, if we allow $\tau_I n\ge (q-1)n$ insertions, assuming that $\mathbb{\Sigma}_q=\{a_1,\cdots, a_q\},$ we can always transform any codeword in $\mathcal{C}$ to a word of length $qn$ obtained by repeating $(a_1,\cdots, a_q)~n$ times. So if $\mathcal{C}$ is list-decodable against $(q-1)n$ insertions with polynomial list size, its size must be at most $L=\mathrm{poly}(n),$ which causes its rate to be asymptotically zero. Hence, to have a positive rate and polynomial list size, we must have $\tau_I<(q-1).$ In the remainder of this work, we always assume $\tau_D<\frac{(q-1)}{q}$ and $\tau_I<(q-1).$

Next, we provide a relation between the list-decodability of codes under the two considered distance.
\begin{lemma}\label{lem:LLDtoILD}
A code $\mathcal{C}$ of length $n$ is $(\tau_I,\tau_D,L)$-insdel-list-decodable whenever it is $(\tau_I+\tau_D,L)$-Levenshtein-list-decodable.
\end{lemma}

\begin{proof}
It is sufficient to show that for any $\mathbf{y}\in \mathbb{\Sigma}_q^N$ such that $N\in[n-\tau_D n,n+\tau_I n],$ we have $\mathcal{B}_{ID}(\mathbf{y},\tau_D n,\tau_I n)\cap\mathcal{C} \subseteq \mathcal{B}_L(\mathbf{y},(\tau_D+\tau_I)n)\cap\mathcal{C}.$ Suppose that $\mathbf{c}\in \mathcal{B}_{ID}(\mathbf{y},\tau_D n,\tau_I n)\cap\mathcal{C}.$ Then there exists $t_D'\le \tau_D n$ and $t_I'\le \tau_I n$ such that $\mathbf{c}$ can be obtained from $\mathbf{y}$ through $t_D'$ deletions and $t_I'$ insertions. By Proposition~\ref{prop:2insdeldist}, we have $d_L(\mathbf{y},\mathbf{c})\le t_I'+t_D'\le (\tau_I+\tau_D)n$ directly implying that $\mathbf{c}\in \mathcal{B}_{L}(\mathbf{y},(\tau_D+\tau_I)n)\cap\mathcal{C},$ completing the proof.
\end{proof}

In this work, we are mainly interested in the lower bound of list-decodability of a code with a given Levenshtein distance. More specifically, we are interested in finding a relation between $\tau_I,\tau_D,L$ and $d$ that ensures that any code of length $n$ with minimum Levenshtein distance $d$ must be $(\tau_I,\tau_D,L)$-insdel-list-decodable.

Lastly we define a function that will be important in our analysis in the latter sections.

\begin{defn}\label{def:rhodeltaLx}
For a relative minimum Levenshtein distance $\delta\in[0,1]$ and list size $L\ge 2,$ we define $\rho^{(\delta,L)}:[1-\delta,1]\rightarrow\mathds{R}$ such that
\[\rho^{(\delta,L)}(x)=\max_{r=1,\cdots, L}\left\{\frac{2L-r+1}{L+1} x -\frac{L}{r}(1-\delta)\right\}.\]
\end{defn}
It is easy to see that $\rho^{(\delta,L)}(x)$ is a piecewise linear function for $x\in[1-\delta,1].$ We refer to Figure~\ref{fig:rhodeltalex} for illustration of the function $\rho^{(\delta, L)}(x).$ Here the term ``turning points'' is defined in the following way. A turning point of the function $\rho^{(\delta,L)}(x)$ is defined as the point where $\rho^{(\delta,L)}(x)$ moves from one linear piece to another. A more complete analysis of $\rho^{(\delta, L)}(x)$ can be found in the Appendix~\ref{app:rhodeltaL}.

\begin{figure}[ht]
     \centering

\includegraphics[scale=0.75]{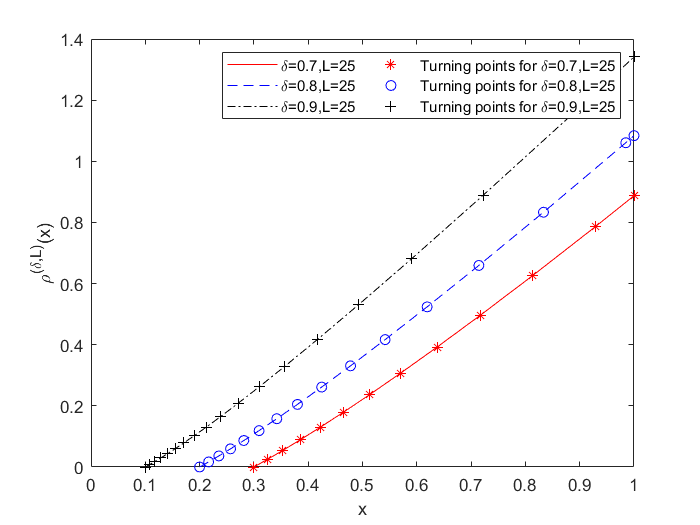}

\caption{Piecewise linear function $\rho^{(\delta,L)}(x)$ for small $L$}
\label{fig:rhodeltalex}
\end{figure}

\section{Insdel List-decodability with List Size 2}\label{sec:Insdel2}
First, we consider the case when $L=2.$

\begin{lemma}\label{lem:L=2}
Let $\mC\subseteq \mathbb{\Sigma}_q^n$ be a $q$-ary  code with length $n$ and minimum Levenshtein distance $d=2\delta n.$ Given $\tau_I$ and $\tau_D,$  if  $\mC$ is not  $(\tau_I,\tau_D, 2)$-insdel-list-decodable, then either $\tau_D\ge\delta$ or
\[\tau_I\ge \rho^{(\delta, 2)}\left(1-\tau_D\right).\]
\end{lemma}

\begin{proof}
Suppose that $\tau_D<\delta.$ For simplicity, we denote $t_I=\tau_I n$ and $t_D=\tau_D n.$ Since $\mC$ is not  $(\tau_I,\tau_D, 2)$-insdel-list-decodable, then there exists an integer $N\in[n-t_D,n+t_I],$ a vector ${\bf y}\in\mathbb{\Sigma}_q^N$ and $3$ distinct codewords of $\mC$ such that ${\bc_0}, {\bc_1}, {\bc_{2}}\in \mB_{ID}(\by, t_D,t_I).$ For any $0\le i\le 2,$ we define $Y_i\subseteq [N]$ to be a set such that $\mathbf{y}|_{Y_i}$ is a longest common subsequence of $\mathbf{y}$ and $\mathbf{c}_i.$ Hence we have $t_D+t_I\ge d_L(\mathbf{y},\mathbf{c}_i)=N+n-2|Y_i|$ or equivalently, $|Y_i|\ge \frac{1}{2}(N+n-t_D-t_I).$ For any $0\le i_1<i_2\le 2,$ we define $Y_{i_1i_2}\triangleq Y_{i_1}\cap Y_{i_2}.$ Similarly, we define $Y_{012}\triangleq Y_0\cap Y_1\cap Y_2.$ It is easy to see that $\mathbf{y}|_{Y_{i_1i_2}}$ is a common subsequence of $\mathbf{c}_{i_1}$ and $\mathbf{c}_{i_2}.$ Hence we must have
$d\le d_L(\mathbf{c}_{i_1},\mathbf{c}_{i_2})\le 2n-2|Y_{i_1i_2}|$
which implies $|Y_{i_1i_2}|\le n-\frac{d}{2}.$ By Inclusion-Exclusion principle, we have

\begin{equation}\label{eq:Yi3}
\left|\bigcup_{i=0}^2 Y_i\right|= \sum_{i=0}^2 |Y_i|-\sum_{0\le i_1<i_2\le 2} |Y_{i_1 i_2}| + |Y_{012}|.
\end{equation}

Since $\bigcup_{i=0}^2 Y_i\subseteq [N]$ and $|Y_{012}|\ge 0,$ we have
\begin{equation*}
N\ge \frac{3}{2}\left(N+n-t_D-t_I\right)-3\left(n-\frac{d}{2}\right).
\end{equation*}
Noting that $N\ge n-t_D,$ we obtain $0\ge  2n-2t_D-\frac{3}{2}t_I-3\left(n-\frac{d}{2}\right)$ which implies
\begin{equation}\label{eq:Zi3}
t_I\ge \frac{4}{3}(n-t_D)-2\left(n-\frac{d}{2}\right).
\end{equation}

Here we again consider the Inclusion-Exclusion principle for the union of $Y_{i_1i_2}.$ We have
\begin{small}
\begin{eqnarray}\label{eq:Yij3}
\nonumber \left|\bigcup_{0\le i_1<i_2\le 2}Y_{i_1i_2}\right|&=& \sum_{0\le i_1<i_2\le 2} |Y_{i_1i_2}|- |Y_{01}\cap Y_{02}|- |Y_{01}\cap Y_{12}|-|Y_{02}\cap Y_{12}|+|Y_{012}|\\
&=&\sum_{0\le i_1<i_2\le 2}|Y_{i_1i_2}|-2|Y_{012}|.
\end{eqnarray}
\end{small}
Noting that $\bigcup_{0\le i_1<i_2\le 2}Y_{i_1i_2}\subseteq [N],$ adding two times Equation~\eqref{eq:Yi3} to Equation~\eqref{eq:Yij3}, we obtain $3N\ge 2\left|\bigcup_{i=0}^2Y_i\right|-\left|\bigcup_{0\le i_1<i_2\le 2}Y_{i_1i_2}\right| \ge 3(N+n-t_I-t_D)-3\left(n-\frac{d}{2}\right).$ Hence we have

\begin{equation}\label{eq:Zij3}
t_I\ge (n-t_D)-\left(n-\frac{d}{2}\right).
\end{equation}

So by Inequalities~\eqref{eq:Zi3} and~\eqref{eq:Zij3}, if $\mathcal{C}$ is not $(t_I,t_D,2)$-list-decodable, then
\begin{equation*}
\tau_I\ge \max\left\{\frac{4}{3}(1-\tau_D)-2\left(1-\delta\right),(1-\tau_D)-\left(1-\delta\right)\right\}=\rho^{(\delta,2)}(1-\tau_D).
\end{equation*}

\end{proof}

Taking the contrapositive of Lemma~\ref{lem:L=2}, we directly have the following result.

\begin{theorem}\label{thm:L=2}
Let $\mathcal{C}\subseteq \mathbb{\Sigma}_q^n$ be a $q$-ary code with length $n$ and minimum Levenshtein distance $d=2\delta n.$ If $\tau_D<\delta$ and $\tau_I <\rho^{(\delta,2)}(1-\delta),$ then $\mathcal{C}$ is $(\tau_I,\tau_D,2)$-insdel-list-decodable.
\end{theorem}

\section{Insdel List-decodability for General List Size}\label{sec:InsL}
In this section, we generalize the result from Section~\ref{sec:Insdel2} to any $L\ge 2.$ More specifically, this section focuses on the proof of the following Theorem.

\begin{theorem}\label{thm:genl}
Let $\mathcal{C}\in\mathbb{\Sigma}_q^n$ be a $q$-ary code of length $n$ with minimum Levenshtein distance $d=2\delta n$ and $L\ge 2$ be a positive integer. We further let $\tau_I$ and $\tau_D$ be two non-negative real numbers such that $\tau_D<\delta$ and $\tau_I<\rho^{(\delta,L)}(1-\tau_D).$
Then $\mathcal{C}$ is $(\tau_I,\tau_D,L)$-insdel-list-decodable.
\end{theorem}

In the remainder of this section, for any positive integers $v\le L,$ sets $Y_0,\cdots, Y_L$ and indices $0\le i_1<i_2<\cdots<i_v\le L,$ we define $Y_{i_1i_2\cdots i_v}\triangleq \bigcap_{j=1}^v Y_{i_j}.$ Before we consider the insdel-list-decodability of codes, first we consider the size of the union of such $Y_{i_1i_2\cdots i_v}$ for any fixed $v.$

Let $S,S_1,\cdots, S_\ell$ be distinct subsets of $\{0,\cdots, L\}.$ We say that $S_1,\cdots, S_\ell$ covers $S$ if we have $\bigcup_{i=1}^\ell S_i=S.$ Furthermore, if $|S_i|=v$ for all $i=1,\cdots, \ell,$ we say that $\{S_1,\cdots, S_\ell\}$ is a $v$-cover of $S.$ Note that the number of such $v$-cover only depends on the size of $S, S_i$ as well as $\ell.$ For any positive integers $j,\ell$ and $v,$ we denote by $A_{j,\ell,v}$ the number of $v$-covers of size $\ell$ of $\{1,\cdots, j\}.$

Note that in general, to obtain the extension from the discussion in Section~\ref{sec:Insdel2}, we use Inclusion-Exclusion principle on $\left|\bigcup_{0\le i_1<\cdots<i_v\le L} Y_{i_1i_2\cdots i_v}\right|$ for increasing values of $v$ to eliminate the terms $|Y_{i_1i_2\cdots i_b}|$ for $b\ge 3$ starting from smaller values of $b,$ which in general has a larger size. Expanding the size of the set $\bigcup_{0\le i_1<\cdots<i_v\le L} Y_{i_1i_2\cdots i_v}$ using Inclusion-Exclusion principle, we can reorder the terms. More specifically, we can see that the formula $\bigcup_{0\le i_1<\cdots<i_v\le L} Y_{i_1i_2\cdots i_v}$ is symmetric with respect to $Y_i.$ Hence, the coefficient of a term in the expanded version of the union only depends on the number of different tuples $(i_1,\cdots, i_v)$ that are intersected to obtain such term. In other words, for any $j\leq L+1$ and $0\le i_1<i_2<\cdots< i_j\le L,$ the coefficient of $Y_{i_1i_2\cdots i_j}$ in the union only depends on the values of $j$ and $v$ while being independent of the actual value of $i_1,\cdots, i_j.$ Furthermore, since each term in the original form of the union is an intersection of $v$ distinct $Y_i$'s, the term in the right hand side can only consist of terms obtained by intersecting at least $v$ distinct $Y_i.$ Hence, denoting such coefficient by $A_{j,v},$ we can reorder the terms in the union to the following form
\begin{equation}\label{eq:YvL}
\left|\bigcup_{0\le i_1<\cdots<i_v\le L} Y_{i_1i_2\cdots i_v}\right|=\sum_{j=v}^{L+1} A_{j,v} \Sigma_j^{(L+1)}
\end{equation}
for some integer $A_{j,v}$ where we define $\Sigma_j^{(L+1)}\triangleq \sum_{0\le i_1<i_2<\cdots<i_{j}\le L} |Y_{i_1i_2\cdots i_{j}}|.$  To calculate $A_{j,v},$ due to the symmetry of any choice of the $j$-tuple $(i_1,\cdots,i_j),$ we only consider the case when $i_t=t-1$ for $t=1,\cdots, j$ and we are interested to find the amount of times $|Y_{012\cdots(j-1)}|$ can be covered by its subsets of size $v.$ Note that for any $v$-cover of size $\ell,$ it contributes $(-1)^{\ell-1}$ to $A_{j,v}.$ This is because such $v$-cover of size $\ell$ is obtained by taking intersection of $\ell$ different sets of size $v.$ So in the Inclusion-Exclusion formula of the size of $\bigcup_{0\le i_1<\cdots<i_v\le L} Y_{i_1i_2\cdots i_v},$ it appears exactly once with coefficient $(-1)^{\ell-1}.$ It can be observed that here $A_{j,v}$ is independent of the original value of $L+1,$ which justifies the omission of the parameter $L$ in the variable $A_{j,v}.$ Furthermore, we have
\begin{equation}\label{eq:AjvtoAjlv}
A_{j,v}=\sum_{\ell=1}^{\binom{j}{v}} (-1)^{\ell-1} A_{j,\ell,v}
\end{equation}
where $\ell$ is at most $\binom{j}{v}$ since there are only $\binom{j}{v}$ subsets of $\{0,\cdots, j-1\}$ of size $v.$

So to simplify Equation~\eqref{eq:YvL}, we are interested in the value of $A_{j,\ell,v}.$ The following result provides a recursive relation of $A_{j,\ell,v}$ with some base cases.

\begin{lemma}\label{lem:Ajlv}
\begin{enumerate}
\item $A_{j,\ell,v}=0$ whenever $\ell> \binom{j}{v}.$ In particular, if $j<v, A_{j,\ell,v}=A_{j,v}=0.$
\item $A_{j,\ell,v}=0$ whenever $\ell v<j.$
\item Suppose that $\ell\le \binom{j}{v}.$ Then
\begin{equation}\label{eq:AjlvRec}
A_{j,\ell,v}=\binom{\binom{j}{v}}{\ell}-\sum_{t=1}^{j-1} \binom{j}{t} A_{t,\ell,v}.
\end{equation}
\end{enumerate}
\end{lemma}
\begin{proof}
\begin{enumerate}
\item Note that a set of size $j$ has exactly $\binom{j}{v}$ pairwise distinct subsets of size exactly $v.$ Hence the maximum number of pairwise distinct sets of size $v$ with its union having size $j$ is $\binom{j}{v}.$ So if $\ell>\binom{j}{v},$ it is impossible to have the union of $\ell$ sets of size $v$ to have size $j,$ implying that $A_{j,\ell,v}=0.$
\item Note that the union of $\ell$ sets, each of size $v$ has size at most $\ell v.$ Hence the union cannot have size larger than $\ell v.$
\item  Let $S_v=\{A\subseteq \{0,1,\cdots, j-1\}:|A|=v\}$ and $\mathcal{T}_{\ell,v}=\{S\subseteq S_v:|S|=\ell\}.$ Then $|\mathcal{T}_{\ell,v}|=\binom{\binom{j}{v}}{\ell}.$ However, note that not all $S\in \mathcal{T}_{\ell,v}$ may result in $\bigcup_{A\in S}A=\{0,\cdots, j-1\}.$ So from $\mathcal{T}_{\ell,v},$ we need to exclude those with union of size less than $j.$ For any $t<j,$ there are $\binom{j}{t}$ subsets of $\{0,\cdots, j-1\}$ of size $t$ and for each $S\subseteq\{0,\cdots, j-1\}$ of size $t,$ there are $A_{t,\ell,v}$ elements of $\mathcal{T}_{\ell,v}$ that covers $S.$ This directly provides Equation~\eqref{eq:AjlvRec}.
\end{enumerate}
\end{proof}

Next, we also provide the values of $A_{j,\ell,v}$ for some special case, which can be easily verifed.

\begin{rmk}\label{rmk:Ajlv}
\begin{enumerate}
\item When $v=1,$
\[A_{j,\ell,1}=\left\{
\begin{array}{cc}
1,&\mathrm{~if~}\ell=j\\
0,&\mathrm{~otherwise}.
\end{array}
\right.\]
In particular, $A_{j,1}=(-1)^{j-1}.$
\item For any $j$ and $v,$ we have $A_{j,\binom{j}{v},v}=1.$ In particular, $A_{j,j}=1.$
\end{enumerate}
\end{rmk}

Next, we provide an explicit value of $A_{j,v}.$ Throughout the remainder of the section, for simplicity of notation, we denote $\Psi_{v}^{(L+1)}\triangleq \left|\bigcup_{0\le i_1<i_2<\cdots<i_v\le L+1}Y_{i_1i_2\cdots i_v}\right|.$

\begin{prop}\label{prop:Ajv}
For any positive integers $j,v$ such that $\max\{2,v\}\le j,$ we have
\begin{equation}\label{eq:Ajv}
A_{j,v}=(-1)^{j-v}\binom{j-1}{v-1}.
\end{equation}
\end{prop}
\begin{proof}
We will prove Equation~\eqref{eq:Ajv} by induction on $j.$ Firstly, we show the claim for $j=v.$ By Remark~\ref{rmk:Ajlv}, $A_{v,v}=1=(-1)^{v-v}\binom{v-1}{v-1}.$

Now we assume that Equation~\eqref{eq:Ajv} is true for any $j'=v,v+1,\cdots, j-1.$ That is, $A_{j',v}=(-1)^{j'-v}\binom{j'-1}{v-1}.$ Then by Equation~\eqref{eq:AjvtoAjlv} and Lemma~\ref{lem:Ajlv}, it can be verified that
\begin{equation}\label{eq:WTS1}
A_{j,v}
=1-\sum_{t=1}^j\binom{j}{t}(-1)^{t-v}\binom{t-1}{t-v}+(-1)^{j-v}\binom{j-1}{v-1}
\end{equation}
For completeness, the proof of Equation~\eqref{eq:WTS1} can be found in Appendix~\ref{app:WTS}.
So to complete our proof, it is sufficient to show the following claim.

\begin{claim}\label{claim:tailgenlv}
For any $v\ge 1$ and $j\ge \max\{v,2\},$
\[\sum_{t=1}^j (-1)^{t-v}\binom{j}{t}\binom{t-1}{v-1}=1.\]
\end{claim}
\begin{proof}
We prove this by induction on both $v$ and $j.$ First, when $v=1,$ the sum can be simplified to $\sum_{t=1}^j(-1)^{t-1}\binom{j}{t}=-\sum_{t=0}^j (-1)^t\binom{j}{t}+1=1$ proving the case when $v=1.$ Next, we show the claim for the case when $j=v$ for any $v\ge 2.$ When $j=v,$ we have $\sum_{t=1}^v(-1)^{t-v}\binom{v}{t}\binom{t-1}{v-1}=\sum_{t=v}^v(-1)^{t-v}\binom{v}{t}\binom{t-1}{v-1}=1$ where the first equality follows from the fact that $\binom{t-1}{v-1}=0$ when $t<v.$

Now we assume that the claim is true for $(v',j')\in\{(v-1,j-1),(v-1,j),(v,j-1)\}$ and we consider the claim for $(v,j).$ That is, we assume that for any such $v'$ and $j', \sum_{t=1}^{j'}(-1)^{t-v'}\binom{j'}{t}\binom{t-1}{v'-1}=1.$ Using the fact that $\binom{a}{b}=\binom{a-1}{b-1}+\binom{a-1}{b}$ for any $0\le a\le b,$  It can be verified that
\begin{eqnarray}\label{eq:WTS2}
\nonumber&&\sum_{t=1}^j (-1)^{t-v}\binom{j}{t}\binom{t-1}{v-1}\\
&=&1-\sum_{t=1}^{j-1}(-1)^{t-v}\binom{j-1}{t}\binom{t-1}{v-1}+\sum_{t=1}^{j-1} (-1)^{t-(v-1)}\binom{j-1}{t}\binom{t-1}{(v-1)-1}=1.
\end{eqnarray}
For completeness, the proof of Equation~\eqref{eq:WTS2} can be found in Appendix~\ref{app:WTS}.
\end{proof}
This completes the proof that $A_{j,v}=(-1)^{j-v}\binom{j-1}{v-1}.$
\end{proof}

We then utilize the explicit expressions of the coefficients of $\Sigma_j^{(L+1)}$ in $\Phi_v^{(L+1)}$ for various $j$ and $v$ to obtain $L$ different lower bounds for $t_I$ with respect to the minimum Levenshtein distance $d,$ the deletion upper bound $t_D$ and list size $L.$ In the following remark, we first discuss the overall strategy of how these $L$ lower bounds are derived.

\begin{rmk}\label{rmk:strat}
For $r=1,\cdots, L,$ our strategy to find the $r$-th lower bound for $t_I$ is to consider a carefully chosen linear combination $\Phi_r^{(L+1)}\triangleq \sum_{u=1}^r c_{r,u}\Psi_u^{(L+1)}$ such that rewriting $\Phi_r^{(L+1)}$ as linear combination of $\Sigma_j^{(L+1)},$ say $\Phi_r^{(L+1)}=\sum_{j=1}^{L+1}\Phi_{r,j}^{(L+1)}\Sigma_j^{(L+1)},$ we have $\Phi_{r,3}^{(L+1)}=\cdots=\Phi_{r,r+1}^{(L+1)}=0.$ In such case, we aim to show that $\Phi_{r,1}^{(L+1)}>0>\Phi_{r,2}^{(L+1)}.$ Lastly, we will also show that $\sum_{j=r+2}^{L+1} \Phi_{r,j}^{(L+1)}\Sigma_j^{(L+1)}\ge 0.$ Note that to show that the sum is non-negative, it is sufficient to show that for any $j\in\{r+2,\cdots, L\}$ such that $j-(r-2)$ is even, $\Phi_{r,j}^{(L+1)}\Sigma_j^{(L+1)} + \Phi_{r,j+1}^{(L+1)}\Sigma_{j+1}^{(L+1)}\ge 0$ and when $L-r$ is odd, $\Phi_{r,L+1}^{(L+1)}\Sigma_{L+1}^{(L+1)}\ge 0.$ Having shown these claims, we have
\begin{equation}\label{eq:Phirgen}
\sum_{u=1}^r c_{r,u} \Psi_u^{(L+1)}\ge \Phi_{r,1}^{(L+1)}\Sigma_1^{(L+1)} + \Phi_{r,2}^{(L+2)}\Sigma_2^{(L+1)}.
\end{equation}

Noting that $\bigcup_{0\le i_1<i_2<\cdots i_v\le L}Y_{i_1i_2\cdots i_v}\subseteq [N], N\ge n-t_D,$ for any $v=1,\cdots, L+1, \Sigma_1^{(L+1)}\ge \frac{L+1}{2}(N+n-t_I-t_D)$ and $\Sigma_2^{(L+1)}\le \frac{(L+1)L}{2}\left(n-\frac{d}{2}\right),$ we have the inequality

\begin{equation*}
\left(\sum_{u=1}^r c_{r,u}\right)N \ge \Phi_{r,1}^{(L+1)}\frac{L+1}{2}(N+n-t_I-t_D)+ \Phi_{r,2}^{(L+1)}\frac{(L+1)L}{2}\left(n-\frac{d}{2}\right)
\end{equation*}

or equivalently
\begin{equation}\label{eq:Phirins}
t_I\ge
\left(\frac{2\Phi_{r,1}^{(L+1)}(L+1)-2\sum_{u=1}^r c_{r,u}}{\Phi_{r,1}^{(L+1)}(L+1)}\right)(n-t_D)
+\frac{\Phi_{r,2}^{(L+1)}}{\Phi_{r,1}^{(L+1)}}L\left(n-\frac{d}{2}\right)
\end{equation}
which provides us with the $r$-th lower bound for $t_I.$
\end{rmk}

As discussed in Remark~\ref{rmk:strat}, fixing $L,$ we also fix $r=1,\cdots, L$ to calculate $\Phi_r^{(L+1)}$ which we use to obtain the $r$-th lower bound for $t_I.$ In the following lemma, we provide a choice of $c_{r,1},\cdots, c_{r,r}$ that satisfy all the requirements stated in Remark~\ref{rmk:strat}.

\begin{lemma}\label{lem:cruchoice}
For $u=1,\cdots, r,$ define $c_{r,u}=r+1-u$ and
\[\Phi_r^{(L+1)}=\sum_{j=1}^{L+1} \Phi_{r,j}^{(L+1)}\Sigma_j^{(L+1)}=\sum_{u=1}^rc_{r,u}\Psi_u^{(L+1)}.\]
Then
\begin{enumerate}
\item $\Phi_{r,1}^{(L+1)}=r>0>-1=\Phi_{r,2}^{(L+1)}.$
\item For $j=3,\cdots, r+1, \Phi_{r,j}^{(L+1)}=0.$
\item For any $j\in\{r+2,\cdots, L\}$ such that $j-r$ is even, $\Phi_{r,j}^{(L+1)}\Sigma_j^{(L+1)} + \Phi_{r,j+1}^{(L+1)}\Sigma_{j+1}^{(L+1)}\ge 0.$ Furthermore, if $L-r$ is odd, $\Phi_{r,L+1}^{(L+1)}\Sigma_{L+1}^{(L+1)}\ge 0.$ This directly implies that $\sum_{j=r+2}^{L+1} \Phi_{r,j}^{(L+1)}\Sigma_j^{(L+1)}\ge 0.$
\end{enumerate}
\end{lemma}
\begin{proof}
Recall that for any $u=1,\cdots, r, \Psi_u^{(L+1)}=\sum_{j=u}^{L+1} (-1)^{j-u}\binom{j-1}{u-1} \Sigma_j^{(L+1)}.$ Hence for $j=1,\cdots, L+1, \Phi_{r,j}^{(L+1)}=\sum_{u=1}^r (r+1-u)(-1)^{j-u}\binom{j-1}{u-1}=\sum_{u=1}^{\min\{r,j\}}(r+1-u)(-1)^{j-u}\binom{j-1}{u-1}.$

\begin{enumerate}
\item When $j=1,$ the sum only has one term, which is when $u=1.$ Hence $\Phi_{r,1}^{(L+1)} = r>0.$ Secondly, when $j=2,$ we may have one or two terms depending on whether $r=1$ or $r>1.$ Note that if $r=1,$ we again only have one term when $u=1$ and hence $\Phi_{1,2}^{(L+1)}=-1<0.$ Next, when $r\ge 2,$ the sum has two terms, when $u=1$ and $2.$ Hence $\Phi_{r,2}^{(L+1)}=-r+(r-1)=-1<0.$ This shows that $\Phi_{r,1}^{(L+1)}=r>0>-1=\Phi_{r,2}^{(L+1)},$ proving the first claim.
\item For the proof of the last two claims, we simplify the expression of $\Phi_{r,j}^{(L+1)}$ by shifting the index from $u$ to $u-1.$ Noting that for any $0<b<a, b\binom{a}{b}=a\binom{a-1}{b-1},$ this gives
\begin{small}
\begin{equation*}
\Phi_{r,j}^{(L+1)}
= r\sum_{u=0}^{\min\{r-1,j-1\}}(-1)^{j-1-u}\binom{j-1}{u}-(j-1)\sum_{u=0}^{\min\{r-2,j-2\}}(-1)^{j-2-u}\binom{j-2}{u}.
\end{equation*}
\end{small}
Recall that for any positive integer $b, \sum_{a=0}^b (-1)^{b-a}\binom{b}{a}=0.$ Hence for any $c<b,$ we have
 $\sum_{a=0}^c (-1)^{b-a}\binom{b}{a}=-\sum_{a=c+1}^b (-1)^{b-a}\binom{b}{a}.$ Now we are ready to prove the last two claims. First, we consider the case when $j\le r.$ Then in this case, both sums in $\Phi_{r,j}^{(L+1)}$ equal zero, which directly implies $\Phi_{r,j}^{(L+1)}=0$ for $3\le j\le r.$ Next, assume that $j=r+1.$ It is then easy to see that $\Phi_{r,j}^{(L+1)}=-r-(-r)=0.$ This completes the proof for the second claim.
\item Lastly, assume that $r+2\le j\le L+1.$ Then we have
\begin{equation*}
\Phi_{r,j}^{(L+1)}=-r\sum_{u=r}^{j-1} (-1)^{j-1-u}\binom{j-1}{u}+(j-1)\sum_{u=r-1}^{j-2}(-1)^{j-u}\binom{j-2}{u}.
\end{equation*}
We denote the first sum by $A$ and the second sum by $B.$ We consider $A$ and $B$ separately using the fact that $\binom{a}{b}=\binom{a-1}{b-1}+\binom{a-1}{b}.$
\begin{itemize}
\item First, we consider $A.$ Then
\begin{equation*}
-\frac{A}{r}=\sum_{u=r}^{j-1}(-1)^{j-1-u}\binom{j-2}{u-1}+\sum_{u=r}^{j-2}(-1)^{j-1-u}\binom{j-2}{u}= (-1)^{j-r-1}\binom{j-2}{r-1}
\end{equation*}
which implies that $A=r(-1)^{j-r}\binom{j-2}{r-1}.$
\item Next, we consider $B.$ Then
\begin{equation*}
\frac{B}{j-1}=\sum_{u=r-1}^{j-2}(-1)^{j-u}\binom{j-3}{u-1}+\sum_{u=r-1}^{j-3}(-1)^{j-u}\binom{j-3}{u}=(-1)^{j-r-1}\binom{j-3}{r-2}
\end{equation*}
which implies that $B=-(j-1)(-1)^{j-r}\binom{j-3}{r-2}.$
\end{itemize}
Hence

\begin{equation*}
\Phi_{r,j}^{(L+1)}=(-1)^{j-r}\cdot\binom{j-3}{r-2}\cdot\left(\frac{r(j-2)}{r-1}-(j-1)\right).
\end{equation*}

It is easy to see that since $j> r+1, \binom{j-3}{r-2}\frac{r(j-2)}{r-1}-(j-1)> 0.$ Hence we must have $\Psi_{r,j}^{(L+1)}>0$ if $j-r$ is even and it is negative when $j-r$ is odd. In particular, when $L-r$ is odd, when we take $j=L+1,$ we have $L+1-r$ to be even implying $\Psi_{r,L+1}^{(L+1)}>0.$

Lastly, suppose that $ r+2\le j\le L$ such that $j-r$ is even. We consider $\Delta=\Phi_{r,j}^{(L+1)}\Sigma_j^{(L+1)}+\Phi_{r,j+1}^{(L+1)}\Sigma_j^{(L+1)}.$ Then we have
\begin{equation*}
\Delta=\binom{j-3}{r-2}\left(\frac{r(j-2)}{r-1}-(j-1)\right)\Sigma_j^{(L+1)}-\binom{j-2}{r-2}\left(\frac{r(j-1)}{r-1}-j\right)\Sigma_{j+1}^{(L+1)}.
\end{equation*}
Now consider $Y_{i_1\cdots i_{j+1}}$ for some $0\le i_1<i_2<\cdots<i_{j+1}\le L+1.$ Note that in the second term, we are deducting $\binom{j-2}{r-2}\left(\frac{r(j-1)}{r-1}-j\right)$ copies of $|Y_{i_1\cdots i_{j+1}}|.$ Note that there are exactly $j+1$ choices of $0\le i_1'<i_2<\cdots<i'_j\le L$ such that $Y_{i_1\cdots i_{j+1}}\subseteq Y_{i'_1\cdots i'_j}.$ Hence in the first term, we are adding $(j+1)\binom{j-3}{r-2}\left(\frac{r(j-2)}{r-1}-(j-1)\right)$ copies of $|Y_{i_1\cdots i_{j+1}}|.$ Hence
\[
\Delta\ge C\sum_{0\le i_1<\cdots<i_{j+1}\le L+1} |Y_{i_1\cdots i_{j+1}}|
\]
where
\begin{eqnarray*}
C&=&(j+1)\binom{j-3}{r-2}\left(\frac{r(j-2)}{r-1}-(j-1)\right)-\binom{j-2}{r-2}\left(\frac{r(j-1)}{r-1}-j\right)\\
&=&\binom{j-3}{r-1}\left(j-\frac{r-1}{j-r-1}\right)\ge 3>0
\end{eqnarray*}
where the inequalities are due to the fact that $j\ge r+2.$
\end{enumerate}
This completes the proof of Lemma~\ref{lem:cruchoice}.
\end{proof}
Lemma~\ref{lem:cruchoice} shows that by taking $c_{r,u}=r+1-u$ for $u=1,\cdots, r,$ we obtain the desirable linear combination to obtain the $r$-th lower bound for $\tau.$

We are now ready to prove the insdel-list-decodability result for a general list size.
\begin{lemma}\label{lem:genl}
Let $\mathcal{C}\in\mathbb{\Sigma}_q^n$ be a $q$-ary code with length $n$ and minimum Levenshtein distance $d=2\delta n.$ We further let $L\ge 2$ be a non-negative integer and $\tau_I,\tau_D$ be non-negative real numbers. Then if $\mathcal{C}$ is not $(\tau_I,\tau_D,L)$-insdel-list-decodable, then either $\tau_D\ge \delta$ or
\[\tau_I\ge \rho^{(\delta,L)}(1-\tau_D).\]
\end{lemma}
\begin{proof}
Suppose that $\tau_D<\delta.$ For simplicity, we denote by $t_I=\tau_I n$ and $t_D=\tau_D n.$ For $r=1,\cdots, L,$ utilizing $\Phi_r^{(L+1)}$ defined in Lemma~\ref{lem:cruchoice}, since $\sum_{u=1}^{r} c_{r,u}=\frac{r(r+1)}{2},\Phi_{r,1}^{(L+1)}=r$ and $\Phi_{r,2}^{(L+1)}=-1,$ by Equation~\eqref{eq:Phirins},
\[t_I\ge\left(\frac{2L-r+1}{L+1}\right)(n-t_D)-\frac{L}{r}\left(n-\frac{d}{2}\right)\]
as required. Hence
\begin{equation}\label{eq:genlmax}
\tau_I\ge \max_{r=1,\cdots, L}\left\{
\left(\frac{2L-r+1}{L+1}\right)(1-\tau_D)-\frac{L}{r}\left(1-\delta\right)
\right\}=\rho^{(\delta,L)}(1-\tau_D).
\end{equation}

This completes the proof.
\end{proof}
%

It is then easy to see that Theorem~\ref{thm:genl} directly follows from Lemma~\ref{lem:genl}.

\begin{rmk}\label{rmk:ourboundgood}
Note that our bound requires $\tau_D<\delta.$ Recall that the unique decoding bound is defined as $\tau_I+\tau_D<\delta.$ Hence our bound is only meaningful when some insertion error occurs or $\tau_I>0.$ Furthermore, when $1-\tau_D\le \frac{L+1}{L-1}\left(1-\delta\right),$ the requirement we have is $\tau_I+\tau_D<\delta,$ which coincides with the unique decoding bound. Hence for our result to be meaningful, we also need $\tau_D$ to satisfy $1-\tau_D>\frac{L+1}{L-1}\left(1-\delta\right)$ or equivalently, $1<\frac{L+1}{2}\delta-\frac{L-1}{2}\tau_D.$ So in particular, to have the code to be insdel-list-decodable with list size $L,$ we need it to have relative minimum Levenshtein distance larger than $\frac{2}{L+1}.$ Hence, if $\delta=o(n),$ we require $L=\omega(n).$
\end{rmk}

%
%

\section{Comparison with the HY Bound}\label{sec:Inscom}

In this section, we compare the lower bound of the list-decodability of an insdel code $\mathcal{C}\subseteq\mathbb{\Sigma}_q^n$ we have derived in Theorem~\ref{thm:genl} with the best known lower bound. In such case, we compare with the bound derived in \cite{HY20}, which we call HY bound. First, we note that the result presented in \cite[Lemma $1$]{HY20} is dependent on the length $N$ of the received word is fixed. So we focus on the comparison with \cite[Theorem $1$]{HY20}. First, we restate the result presented in \cite[Theorem $1$]{HY20}.

\begin{theorem}[Restatement of Theorem $1$ of \cite{HY20}]\label{thm:HY20}
Let $\mathcal{C}\subseteq \mathbb{\Sigma}_q^n$ be a code of minimum Levenshtein distance $d=2\delta n.$ Define non-negative integers $t_I=\tau_I n$ and $t_D=\tau_D n<n.$ If $\tau_I<\frac{(\delta-\tau_D)(1-\tau_D)}{(1-\delta)}$ and $L= \left\lfloor\frac{\delta(1+\tau_I)}{(\delta-\tau_D)(1-\tau_D)-(1-\delta)\tau_I}\right\rfloor,$ the code $\mathcal{C}$ is $(\tau_I,\tau_D,L)$-insdel-list-decodable.
\end{theorem}

Denote by $\phi_1^{(\delta, L)}(x)\triangleq \frac{x^2}{1-\delta}-x.$ Note that by the value assigned to $L,$ we have $L>\frac{\delta(1+\tau_I)}{(\delta-\tau_D)(1-\tau_D)-(1-\delta)\tau_I}-1$ or equivalently,  $\tau_I<\phi_2^{(\delta,L)}(1-\tau_D)$
where
\begin{small}
\[\phi_2^{(\delta,L)}(x)\triangleq\frac{(L+1)x^2-(L+1)(1-\delta)x+(1-\delta)-1}{L(1-\delta)+1}.\]
\end{small}
Hence Theorem~\ref{thm:HY20} can be reformulated as $\tau_I<\min\{\phi_1^{(\delta,L)}(1- \tau_D),\phi_2^{(\delta,L)}(1- \tau_D)\}.$ Now, for given $\delta$ and $L,$ we are interested to show that our bound outperforms HY bound for some values of $\tau_D.$ Note that for such values of $\tau_D,$ we have $\min\{\phi_1^{(\delta,L)}(1-\tau_D),\phi_2^{(\delta,L)}(1-\tau_D)\}<\rho^{(\delta,L)}(1-\tau_D).$ Recall that by Remark~\ref{rmk:ourboundgood}, our bound coincides with the unique decoding bound when $\tau_D\ge 1-\frac{L+1}{L-1}(1-\delta).$ Hence, in addition to the requirement that $\min\{\phi_1^{(\delta,L)}(1-\tau_D),\phi_2^{(\delta,L)}(1-\tau_D)\}<\rho^{(\delta,L)}(1-\tau_D),$ we also require that $\tau_D<1-\frac{L+1}{L-1}(1-\delta).$ In order for such requirement to make sense, we require that $1-\delta<\frac{L-1}{L+1}$ or equivalently $\delta>\frac{2}{L+1}.$
Indeed, the existence of such $\tau_D$ is discussed in Theorem~\ref{thm:oursbettergen}.

%

\begin{theorem}\label{thm:oursbettergen}
Let $\delta$ and $L$ be given fixed constants. Then, there exists $0<\delta_1<1$ such that if $\delta_1<\delta<1,$ there exists an open interval $\mathcal{I}^{(\delta,L)}\subseteq \left[0,1-\frac{L+1}{L-1}(1-\delta)\right)$ where for any $\tau_D\in \mathcal{I}^{(\delta,L)}, \min\{\phi_1^{(\delta,L)}(1-\tau_D),\phi_2^{(\delta,L)}(1-\tau_D)\}<\rho^{(\delta,L)}(1-\tau_D).$
%
%
\end{theorem}
\begin{proof}

Note that by a simple algebraic manipulation, we have $\phi_2^{(\delta,L)}(x)<\phi_1^{(\delta,L)}(x)$ if and only if $x^2-(1-\delta)x+(1-\delta)>0.$ Note that for any $x,$ we have $x^2-(1-\delta)x+(1-\delta)=\left(x-\frac{1-\delta}{2}\right)^2+\frac{1-\delta}{4}(3+\delta)>0$ since $\delta\in(0,1).$

 This shows that for any $\tau_D<\delta,$ we have $\phi_2^{(\delta,L)} (1-\tau_D)<\phi_1^{(\delta,L)} (1-\tau_D).$ Hence, Theorem~\ref{thm:HY20} can be further simplified to $\tau_I<\phi_2^{(\delta,L)} (1-\tau_D)$ and we only aim to show that $\phi_2^{(\delta,L)}(1-\tau_D)<\rho^{(\delta,L)}(1-\tau_D).$

Recall that $\rho^{(\delta,L)}(x)$ is a continuous piecewise function for $x\in[1-\delta,1].$ Furthermore, as has been observed in Appendix~\ref{app:rhodeltaL}, we have the following expression for $\rho^{(\delta,L)}(x).$
\[
\rho^{(\delta,L)}(x)=\left\{
\begin{array}{cc}
x-(1-\delta),&\mathrm{~if~}(1-\delta)\leq x\leq \frac{L+1}{L-1}(1-\delta)\\
\frac{L+2}{L+1}x-\frac{L}{L-1}(1-\delta),&\mathrm{~if~}\frac{L+1}{L-1}(1-\delta)<x\leq \frac{L(L+1)}{(L-1)(L-2)}(1-\delta)
\end{array}
\right..
\]
It is then easy to see that when $\tau_D\rightarrow 1-\frac{L+1}{L-1}(1-\delta),$
\[ \rho^{(\delta,L)}(1-\tau_D)\rightarrow \rho^{(\delta,L)}\left(\frac{L+1}{L-1}(1-\delta)\right)= \frac{2}{L-1}(1-\delta)\] while $\phi_2^{(\delta,L)}(1-\tau_D)$ approaches

\begin{equation*}
\phi_2^{(\delta,L)}\left(\frac{L+1}{L-1}(1-\delta)\right)= \frac{1}{L(1-\delta)+1}\left[\frac{2(L+1)^2}{(L-1)^2}(1-\delta)^2+(1-\delta)-1\right].
\end{equation*}

It can be verified that $\phi_2^{(\delta,L)}\left(\frac{L+1}{L-1}(1-\delta)\right)<\rho^{(\delta,L)}\left(\frac{L+1}{L-1}(1-\delta)\right)$ if and only if $(6L+2)(1-\delta)^2+(L^2-4L+3)(1-\delta)-(L^2-2L+1)<0,$ which has one positive $\beta_2>0$ and one negative real roots $\beta_1<0.$ Hence $\phi_2^{(\delta,L)}\left(\frac{L+1}{L-1}(1-\delta)\right)<\rho^{(\delta,L)}\left(\frac{L+1}{L-1}(1-\delta)\right)$ if and only if $\beta_1<1-\delta<\beta_2$ or equivalently $1-\beta_2<\delta<1-\beta_1.$ Noting that $\beta_1<0$ and we require $\delta>\frac{2}{L+1},$ setting $\delta_1=\max\left\{\frac{2}{L+1},1-\beta_2\right\},$ for any $\delta_1<\delta<1,$
\[\phi_2^{(\delta,L)}\left(\frac{L+1}{L-1}(1-\delta)\right)<\rho^{(\delta,L)}\left(\frac{L+1}{L-1}(1-\delta)\right)\]

Since both functions are continuous, it is easy to see that for any $L$ and $\delta\in(\delta_1,1),$ there exists an open interval $\mathcal{I}^{(\delta,L)}\subseteq \left[0,1-\frac{L+1}{L-1}(1-\delta)\right)$ such that for any $\tau_D\in \mathcal{I}^{(\delta,L)}, \rho^{(\delta,L)}(1-\tau_D)>\phi_2^{(\delta,L)}(1-\tau_D),$ proving the claim.
\end{proof}


\begin{rmk}\label{rmk:valueofbeta}
A simple algebraic manipulation yields
\[\beta_2=\frac{L-1}{4(3L+1)}\left(-(L-3)+\sqrt{L^2+18L+17}\right).\] It is also easy to verify that
\[\frac{L-1}{4(3L+1)}\left(-(L-3)+\sqrt{L^2+18L+17}\right)<\frac{L-1}{L+1}.\]
Hence, we have \[\delta_1=1-\beta_2=\frac{L^2+8L+7-(L-1)\sqrt{L^2+18L+17}}{4(3L+1)}.\]

%
\end{rmk}

To better illustrate the region claimed in Theorem~\ref{thm:oursbettergen}, we provide an example when we set $L=2,$ which can be found in Example~\ref{ex:oursbettergen}.

\begin{ex}\label{ex:oursbettergen}
Note that when $L=2,$ when $x>3(1-\delta),$ we have $\rho^{(\delta,2)}(x)=\frac{4}{3}x-2(1-\delta)$ while $\phi_2^{(\delta,2)}(x)=\frac{3}{2(1-\delta)+1}x^2-\frac{3(1-\delta)}{2(1-\delta)+1}x+\frac{(1-\delta)-1}{L(1-\delta)+1}.$
It can be verified that $\rho^{(\delta,2)}(1-\tau_D)>\phi_2^{(\delta,2)}(1-\tau_D)$ if and only if $\delta>\frac{27-\sqrt{57}}{28}$ and $1-\tau_D\in\left(3(1-\delta),\alpha\right)$
where
\begin{tiny}
\[\alpha=\frac{17(1-\delta)+4+\sqrt{-143(1-\delta)^2-188(1-\delta)+124}}{18}.\]
\end{tiny}
As presented in Section~\ref{sec:IntComp}, we illustrate this example in Figure~\ref{fig:rhodelta2exalt}.
\end{ex}

\section{List-decodability of Various Insdel Codes}\label{sec:inscons}

In this section, we utilize the result derived in Theorem~\ref{thm:genl} to determine a lower bound for the insdel-list-decodability of various families of insdel codes.

\subsection{Reed-Solomon Codes}\label{sec:RScodes}
The first family we consider is the Reed-Solomon codes, which is one of the most commonly used codes under Hamming metric. First, we recall the definition of a Reed-Solomon code.
\begin{defn}\label{def:RS}
Let $q$ be a prime power and $\F_q$ be the finite field of $q$ elements. Let $n\le q$ be a positive integer and $\alpha_1,\cdots, \alpha_n$ be $n$ distinct elements of $\F_q.$ Denote ${\boldsymbol \alpha}=(\alpha_1,\cdots, \alpha_n).$ For a positive integer $k\in\{1,\cdots, n\},$ we define $\F_q[x]_{<k},$ the set of all polymomials over $\F_q$ with degree less than $k.$ We define the Reed-Solomon code with dimension $k$ and evaluation vector ${\boldsymbol\alpha}, \mathtt{RS}_{\boldsymbol\alpha}(n,k)$ as follows
\[\mathtt{RS}_{\boldsymbol\alpha}(n,k)\triangleq\left\{(f(\alpha_1),\cdots, f(\alpha_n)):f(x)\in \F_q[x]_{<k}\right\}.\]
\end{defn}

It can be shown that regardless of the choice of ${\boldsymbol \alpha},$ the minimum Hamming distance of $\mathtt{RS}_{\boldsymbol\alpha}(n-k)$ is always $n-k+1.$ However, when considering its Levenshtein distance, the minimum Levenshtein distance of $\mathtt{RS}_{\boldsymbol \alpha}(n-k)$ is not invariant with respect to ${\boldsymbol \alpha}.$ On one hand, if ${\boldsymbol \alpha}$ is chosen such that the corresponding Reed-Solomon code is cyclic, then its minimum Levenshtein distance is $2.$ On the other hand, it can also be shown that for some choices of ${\boldsymbol \alpha},$ the Reed-Solomon code of dimension $k$ can correct up to $n-2k+1$ Levenshtein errors~\cite{CST21}. This implies that for such choice of the evaluation points, the minimum Levenshtein distance is at least $2n-4k+4.$ Although there have been numerous works on the insdel-correcting capabilities of Reed-Solomon codes (see, for example~\cite{DLTX21, LT21, CZ22, LX21, CST21}), there have not been any study on its list-decoding capability. It is easy to see from Theorem~\ref{thm:genl} that the upper bound for $t_I$ increases as the minimum Levenshtein distance increases. By Theorem~\ref{thm:genl}, we obtain the following result.

\begin{theorem}[List-Decodability of some Reed-Solomon codes]\label{thm:RSLD}
Let $q$ be a prime power, $n\le q$ and $k=Rn$ be positive integers for some $R\in(0,1).$ Furthermore, let $\delta=1-2R.$ Then if $\tau_D<1-2R$ and $\tau_I<\rho^{(\delta,L)}(1-\tau_D),$ there exists ${\boldsymbol \alpha}=(\alpha_1,\cdots, \alpha_n)\in \mathds{F}_q^n,$ a vector of $n$ distinct elements of $\F_q$ such that $\mathtt{RS}_{\boldsymbol \alpha}(n,k)$ is $(\tau_I,\tau_D,L)$-insdel-list-decodable.
\end{theorem}

An illustration of the region of $(\tau_D,\tau_I)$ that can be list-decoded with list size $25$ by Reed-Solomon codes of various rate $R$ can be observed in Figure~\ref{fig:rhodeltalRS}.

\begin{figure}[ht]
     \centering

\includegraphics[scale=0.75]{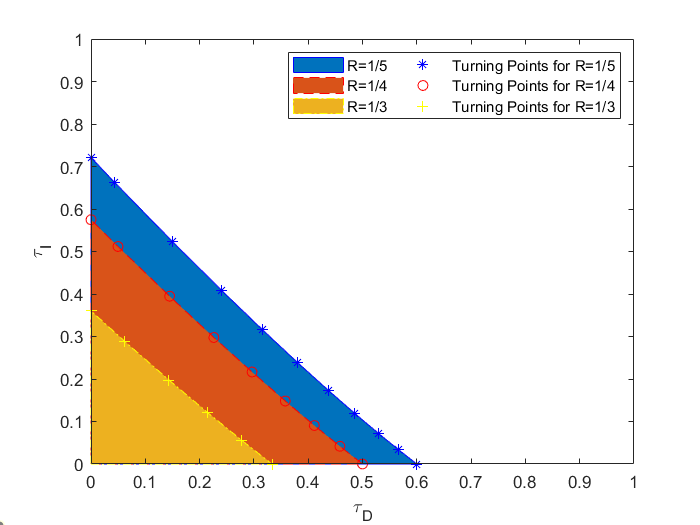}

\caption{Region of relative insertion and deletion errors list-decodable by Reed-Solomon codes of various rates $R$ with list size $L=25$}
\label{fig:rhodeltalRS}
\end{figure}


\subsection{Varshamov-Tenengolts Codes}\label{sec:VTcodes}
In this remainder of this section, for any integer $q\ge 2,$ we define $\mathbb{\Sigma}_q\triangleq\{0,1,\cdots, q-1\}\subseteq\mathds{Z}.$
The next family we consider is the Varshamov-Tenengolts Codes or VT codes for short. It was first constructed over alphabet of size $2$ in~\cite{VT65} and shown to be able to correct up to single insertion or deletion errors in~\cite{Lev65}. Here we recall the definition of a binary VT codes.

\begin{defn}\label{def:VT2}
For $a\in\{0,\cdots, n\},$ the binary VT code $\mathcal{VT}_a(n)$ is defined to be
\[\mathcal{VT}_a(n)\triangleq \left\{
\begin{array}{c}
\mathbf{c}=(c_1,\cdots, c_n)\in \mathbb{\Sigma}_2^n:\\
\sum_{i=1}^n i\cdot c_i\equiv a\pmod{n+1}
\end{array}
\right\}.\]
\end{defn}

Such construction is then extended by Tenengolts in~\cite{Ten84} to any alphabet of size at least $2.$ We recall the definition of a non-binary VT codes.

\begin{defn}\label{def:VTq}
Let $q>2$ be a positive integer. For any $q$-ary vector $\mathbf{s}=(s_0,\cdots, s_{n-1})\in \mathbb{\Sigma}_q^n,$ define a corresponding length $(n-1)$ binary vector $\mathbf{a}_{\mathbf{s}}=(\alpha_1,\cdots, \alpha_{n-1})\in \mathbb{\Sigma}_2^{n-1}$ such that for $1\le i\le n-1,$
\[\alpha_i=\left\{
\begin{aligned}
1,&\mathrm{~if~} s_i\ge s_{i-1}\\
0,&\mathrm{~if~} s_i<s_{i-1}.
\end{aligned}
\right.\]

Then for any $a=0,\cdots, n-1$ and $b\in\mathbb{\Sigma}_q,$ the $q$-ary VT code $\mathcal{VT}^{(q)}_{a,b}(n)$ is defined to be
\[\mathcal{VT}^{(q)}_{a,b}(n)\triangleq \left\{
\begin{array}{c}
\mathbf{s}=(s_0,\cdots, s_{n-1}):\\
 \sum_{i=1}^{n-1} i\alpha_i\equiv a\pmod{n}\\
\mathrm{~and~}\sum_{i=0}^{n-1} s_i\equiv b\pmod{q}
\end{array}
\right\}.\]
\end{defn}

Since both $\mathcal{VT}_a(n)$ and $\mathcal{VT}^{(q)}_{a,b}(n)$ can correct a single insertion/deletion error, their minimum Levenshtein distance must be at least $3.$ Furthermore, since their minimum Levenshtein distance must also be even, we have that their minimum Levenshtein distance to be at least $4.$
Furthermore, it can also be verified that when $n$ is sufficiently large, $\mathcal{VT}_a(n)$ and $\mathcal{VT}^{(q)}_{a,b}(n)$ have minimum Levenshtein distance of exactly $4.$
Hence, by Theorem~\ref{thm:genl}, we obtain the following result.

\begin{theorem}[List-Decodability of VT code]\label{thm:VTLD}
Let $q,n\ge 2$ be positive integers. If $\tau_D<\frac{2}{n}$ and $\tau_I<\rho^{\left(\frac{4}{n},L\right)}(1-\tau_D),$ then for any $a=0,\cdots, n, a'=0,\cdots, n-1$ and $b\in \mathbb{\Sigma}_q,$ both $\mathcal{VT}_a(n)$ and $\mathcal{VT}_{a',b}^{(q)}(n)$ are $(\tau_I,\tau_D,L)$-insdel-list-decodable.
\end{theorem}
To the best of our knowledge, there has only been one work on the list-decoding of binary VT codes~\cite{Wac18} which is only against deletions errors. Recall that for our upper bound $\rho^{\left(\frac{4}{n},L\right)}$ to be meaningful, we require $\tau_I>0$ and $n<L+1-\frac{L-1}{2}\tau_D n\le L+1$ where $\tau_D\le \frac{1}{n}.$ However, when such requirement is satisfied, both codes are insdel-list-decodable against up to $\tau_I n>0$ insertions and $\tau_D n=1$ deletion.


\subsection{Helberg Code}\label{sec:HC}
Note that VT codes are only constructed to correct a single insertion/deletion error. This construction was then extended in~\cite{HF02}, which is more commonly known as binary Helberg codes. Such family is shown to be capable of correcting up to $s$ insertion/deletion errors where $s$ is set to be one of its parameters. The family of Helberg codes is then generalized to non-binary alphabets in~\cite{LN16}, which is again shown to be capable of correcting up to $s$ insertion/deletion errors for some parameter $s.$ Here we recall the definition of a Helberg code.

\begin{defn}\label{def:HC}
Let $s<n$ and $q\ge 2$ be positive integers. Define a sequence of non-negative integer $\{v_i(q,s)\}_{i\in \mathbb{Z}}$ where
\[v_i(q,s)=\left\{
\begin{aligned}
0,&\mathrm{~if~}i\le 0\\
1+(q-1)\sum_{j=1}^s v_{i-j}(q,s),&\mathrm{~if~}i\ge 1.
\end{aligned}
\right.
\]
Define $m\ge v_{n+1}(q,s)=1+\sum_{j=0}^{s-1}v_{n-j}(q,s)$ and $a\in\{0,\cdots, m-1\}.$ Then we define the $q$-ary Helberg code $\mathcal{C}_H(q,n,s,a)$ as follows:
\begin{small}
\[\mathcal{C}_H(q,n,s,a)\triangleq \left\{
\begin{array}{c}
(x_1,\cdots, x_n)\in \mathbb{\Sigma}_q^n:\\
\sum_{i=1}^n v_i(q,s) x_i\equiv a\pmod{m}
\end{array}
\right\}.\]
\end{small}
\end{defn}

Note that since $\mathcal{C}_H(q,n,s,a)$ can correct up to $s$ insertion/deletion errors, their minimum Levenshtein distance must be at least $2s+1.$ Furthermore, since its minimum Levenshtein distance must also be even, we have that its minimum Levenshtein distance to be at least $2s+2.$ We note that here $s$ is a constant with respect to $n.$ Hence our lower bound for $d_L(\mathcal{C}_H(q,n,s,a))$ is also a constant with respect to $n.$ Hence we require $n<\frac{L+1}{4}(2s+2)-\frac{L-1}{2} t_D$ where $t_D<s+1.$ In other words, for our bound to be meaningful, we again require $L=\Omega(n).$ In such case, we have the following result.

\begin{theorem}[List-Decodability of the Helberg code]\label{thm:HCLD}
Let $s<n,q,L\ge 2$ and $a\in\{0,\cdots, m-1\}$ for some $m\ge v_{n+1}(q,s)$ be positive integers. Then if $\tau_D<\frac{s+1}{n}$ and $\tau_I<\rho^{\left(\frac{2s+2}{n},L\right)}(1-\tau_D),$ the Helberg code $\mathcal{C}_H(q,n,s,a)$ is $(\tau_I,\tau_D,L)$-insdel-list-decodable.
\end{theorem}
 To the best of our knowledge, there has not been any study on the list-decoding of Helberg code. 
 
\subsection{Optimal Codes} The next code we consider is the $s$-deletion correcting codes proposed by Sima and Bruck~\cite{SB21} which is denoted by SB code and the two deletion correcting codes proposed by Guruswami and H{\aa}stad~\cite{GH21} which we denote by GH code. Due to the complexity of the construction which uses synchronization vector, dense sequences and hash function for the case of SB code and the use of sketches and regular strings for the case of GH code, combined with the fact that in our analysis, we only require their minimum Levenshtein distances, we omit the definition of such codes. As discussed before, in either case, we use the deletion correcting capability shown for the codes to provide a lower bound for their minimum Levenshtein distance. Noting that both codes only works for minimum Levenshtein distance that is constant with respect to their lengths ($s$ for some constant $s$ in the case of SB codes and $2$ for the case of GH code), in order for our bound to be meaningful, we require $L=\Omega(n).$ Then we have the following two results.

\begin{theorem}[List-Decodability of SB code~\cite{SB21}]\label{thm:SBCLD}
Let $s<N$ and $n=N+8s\log N+o(\log N)$ while $\mathcal{C}$ be the $s$-deletion correcting SB code of length $n$ constructed in~\cite[Theorem $1$]{SB21}. Let $L=\Omega(n)$ be the considered list size. Then if $\tau_D<\frac{s+1}{n}$ and $\tau_I<\rho^{\left(\frac{2s+2}{n},L\right)}(1-\tau_D),$ the SB code $\mathcal{C}$ is $(\tau_I,\tau_D,L)$-insdel-list-decodable. 
\end{theorem}

\begin{theorem}[List-Decodability of GH code~\cite{GH21}]\label{thm:GHCLD}
Let $\mathcal{C}$ be a $2$-deletion correcting code constructed in~\cite[Theorem $1.1$]{GH21} of length $n$ and let $L=\Omega(n)$ be the list size we consider. Then if $\tau_D<\frac{3}{n}$ and $\tau_I<\rho^{\left(\frac{6}{n},L\right)}(1-\tau_D),$ the GH code $\mathcal{C}$ is $(\tau_I,\tau_D,L)$-insdel-list-decodable.
\end{theorem}

\begin{appendices}

\section{Discussion on Our Bound}\label{app:rhodeltaL}
In the following lemma, we show that $\rho^{(\delta,L)}(x)$ consists of at most $L$ linear pieces and $\rho^{(\delta,L)}(x)>0.$
\begin{lemma}\label{lem:rhovalact}
Let $\delta$ and $L$ be given and $\rho^{(\delta,L)}$ be as defined above. Let $r_{\min}\in[L]$ such that
\[r_{\min}\triangleq\min\left\{
r\in\{1,2,\cdots, L\}:
 \frac{L(L+1)}{r(r+1)}\left(1-\delta\right)<1
\right\}.\]
Then $\rho^{(\delta,L)}(x)>0$  and
$\rho^{(\delta,L)}(x)$ is a piecewise linear function in $[1-\delta,1]$ with $L-r_{\min}+1$ pieces where
\begin{enumerate}
\item When $1-\delta\le x\le  \frac{L+1}{L-1}(1-\delta),$ we have $\rho^{(\delta,L)}(x)=x-(1-\delta).$
\item For $r=r_{\min}+1,\cdots, L-1,$ when $\frac{L(L+1)}{r(r+1)}(1-\delta)<x\le \frac{L(L+1)}{r(r-1)}(1-\delta),$ we have $\rho^{(\delta,L)}(x) = \frac{2L-r+1}{L+1}x -\frac{L}{r}(1-\delta).$
\item Lastly, when $\frac{L(L+1)}{r_{\min}(r_{\min}+1)}(1-\delta)<x\le 1,$ we have $\rho^{(\delta,L)}(x)=\frac{2L-r_{\min}+1}{L+1} x -\frac{L}{r_{\min}}(1-\delta).$
\end{enumerate}
%

\end{lemma}
\begin{proof}
First, we consider the definition $\rho^{(\delta,L)}(x)$ where the domain is the whole non-negative real space.
For $r=1,\cdots, L,$ define $u_r(x)=\frac{2L-r+1}{L+1}x-\frac{L}{r}\left(1-\delta\right).$ By simple algebraic manipulation, for $r=1,\cdots, L-1,$ we obtain $u_r(x)>u_{r+1}(x)$ if and only if $x>\frac{L(L+1)}{r(r+1)}\left(1-\delta\right).$ This directly implies that
\begin{enumerate}
\item When $0\le x\le \frac{L+1}{L-1}(1-\delta),$ we have $\rho^{(\delta,L)}(x)=x-(1-\delta).$
\item For $r=2,\cdots, L-1,$ when $\frac{L(L+1)}{r(r+1)}(1-\delta)<x\le \frac{L(L+1)}{r(r-1)}(1-\delta),$ we have $\rho^{(\delta,L)}(x)=\frac{2L-r+1}{L+1}x-\frac{L}{r}(1-\delta).$
\item Lastly, when $x>\frac{L(L+1)}{2}(1-\delta),$ we have $\rho^{(\delta,L)}(x)=\frac{2L}{L+1}x - L(1-\delta).$ 
\end{enumerate}

%
%
%
%
%
%

Now we show the positivity of $\rho^{(\delta,L)}$ in all the $L$ different intervals.
\begin{enumerate}
\item Note that when $x>\frac{L+1}{L-1}\left(1-\delta\right), \rho^{(\delta,L)}(x)\ge \frac{L+2}{L+1}x-\frac{L}{L-1}\left(1-\delta\right)>0.$ This shows that $\rho^{(\delta,L)}(x)>0$ when $x>\frac{L+1}{L-1}\left(1-\delta\right).$
\item Lastly, suppose that $0\le x\le \frac{L+1}{L-1}\left(1-\delta\right).$ Then $\rho^{(\delta,L)}(x)=x-\left(1-\delta\right).$ It is then easy to see that $\rho^{(\delta,L)}(x)>0$ if and only if $x>1-\delta.$
\end{enumerate}

Note that since $1-\delta<1, r_{\min}\le L.$ Hence we obtain the claimed value of $\rho^{(\delta,L)}(x),$ completing the proof.
\end{proof}


\section{Proofs of Equations}\label{app:WTS}
Here, we provide the proof of some calculation claims we made for completeness.
\begin{enumerate}
\item \textbf{Proof of Equation~\eqref{eq:WTS1}.} Recall that $A_{j,v}=\sum_{\ell=1}^{\binom{j}{v}} (-1)^{\ell-1} A_{j,\ell,v}, A_{j,\ell,v}=\binom{\binom{j}{v}}{\ell}-\sum_{t=1}^{j-1}\binom{j}{t} A_{t,\ell,v}$ and we assumed that for $t\le j-1, A_{t,v}=(-1)^{t-v}\binom{t-1}{v-1}.$ Then
\begin{eqnarray*}
A_{j,v}&=&\sum_{\ell=1}^{\binom{j}{v}} (-1)^{\ell-1} A_{j,\ell,v}=-\sum_{\ell=1}^{\binom{j}{v}}(-1)^{\ell}\binom{\binom{j}{v}}{\ell} -\sum_{t=1}^{j-1}\binom{j}{t} \sum_{\ell=1}^{\binom{j}{v}}(-1)^{\ell-1} A_{t,\ell,v}\\
&=&-\left(\sum_{\ell=0}^{\binom{j}{v}}(-1)^\ell\binom{\binom{j}{v}}{\ell}-1\right)-\sum_{t=1}^{j-1}\binom{j}{t}\sum_{\ell=1}^{\binom{t}{v}}(-1)^{\ell-1} A_{t,\ell,v}\\
&=&1-\sum_{t=1}^{j-1}\binom{j}{t}A_{t,v}=1-\sum_{t=1}^{j-1}\binom{j}{t}(-1)^{t-v}\binom{t-1}{v-1}\\
&=&1-\sum_{t=1}^j\binom{j}{t}(-1)^{t-v}\binom{t-1}{t-v}+(-1)^{j-v}\binom{j-1}{v-1}
\end{eqnarray*}
\noindent where the third equality is due to the fact that $A_{t,\ell,v}=0$ for any $\ell>\binom{t}{v}$ while the equalities in the subsequent line is due to Equation~\eqref{eq:AjvtoAjlv} and the induction hypothesis. \qed
\item \textbf{Proof of Equation~\eqref{eq:WTS2}.} Equation~\eqref{eq:WTS2} can be verified using the fact that $\binom{a}{b}=\binom{a-1}{b-1}+\binom{a-1}{b},$ which can be found below.

\begin{eqnarray*}
\sum_{t=1}^j (-1)^{t-v}\binom{j}{t}\binom{t-1}{v-1}&=&\sum_{t=v}^j(-1)^{t-v}\binom{j}{t}\binom{t-1}{v-1}\\
&=&\sum_{t=v}^j (-1)^{t-v}\binom{j-1}{t}\binom{t-1}{v-1}+\sum_{t=v}^j (-1)^{t-v} \binom{j-1}{t-1}\binom{t-1}{v-1}\\
&=&\sum_{t=1}^{j-1} (-1)^{t-v}\binom{j-1}{t}\binom{t-1}{v-1}+\sum_{t=v}^j(-1)^{t-v}\binom{j-1}{t-1}\binom{t-2}{v-1}\\
&&+\sum_{t=v}^j (-1)^{t-v}\binom{j-1}{t-1}\binom{t-2}{v-2}\\
&=&1-\sum_{t=v-1}^{j-1}(-1)^{t-v}\binom{j-1}{t}\binom{t-1}{v-1}\\
&&+\sum_{t=v-1}^{j-1} (-1)^{t-(v-1)}\binom{j-1}{t}\binom{t-1}{(v-1)-1}\\
&=&1-\sum_{t=1}^{j-1}(-1)^{t-v}\binom{j-1}{t}\binom{t-1}{v-1}\\
&&+\sum_{t=1}^{j-1} (-1)^{t-(v-1)}\binom{j-1}{t}\binom{t-1}{(v-1)-1}=1.
\end{eqnarray*}
Here we note that the first and the fifth equality are due to the fact that for any pair of non-negative integer $a$ and $b,$ we have $\binom{a}{b}=0$ for any $a<b.$ Furthermore, the fourth equality is obtained by the following. First, note that by the induction assumption, we have $\sum_{t=1}^{j-1}(-1)^{t-v}\binom{j-1}{t}\binom{t-1}{v-1}=1.$ This implies that the first term in the third equality equals to $1.$ The other two terms are modified by using the substitution $t'=t-1$ and relabelling the variable for the sum back to $t.$ Lastly, the last equality is again due to the induction assumption, namely $\sum_{t=1}^{j-1}(-1)^{t-(v-1)}\binom{j-1}{t}\binom{t-1}{(v-1)-1}=\sum_{t=1}^{j-1}(-1)^{t-v}\binom{j-1}{t}\binom{t-1}{v-1}=1.$

\end{enumerate}

\end{appendices}

\end{document}